\documentclass[letter,11pt]{article}

\usepackage{amsmath}
\usepackage{amsthm}
\usepackage[shortlabels]{enumitem}
\usepackage[margin=1in]{geometry}
\usepackage{graphicx}
\usepackage{newtxmath}
\usepackage{prettyref}
\usepackage{tikz}
\usepackage{xcolor}

\usepackage{hyperref}% Should be last usepackage

\DeclareMathAlphabet\euscr{U}{eus}{m}{n}

\newcommand{\bl}{\bar\ell}
\newcommand{\br}{\bar r}
\newcommand{\bx}{\bar x}
\newcommand{\by}{\bar y}

\newcommand{\cL}{\mathcal L}
\newcommand{\Conf}{\euscr{C}}
\newcommand{\cQ}{\mathcal Q}
\newcommand{\cW}{\mathcal W}
\DeclareMathOperator{\cl}{clique}
\newcommand{\Ind}{\vvmathbb 1}
\let\1\Ind
\DeclareMathOperator{\supp}{supp}
\newcommand{\N}{\mathbb N}
\newcommand{\R}{\mathbb R}
\newcommand{\hPhi}{\hat{\Phi}}
\newcommand{\tPhi}{\tilde{\Phi}}
\newcommand{\hw}{\hat w}
\newcommand{\tw}{\tilde w}

\let\pref=\prettyref

\newrefformat{cl}{Claim~\ref{#1}}
\newrefformat{lem}{Lemma~\ref{#1}}

\newtheorem{theorem}{Theorem}
\newtheorem{claim}[theorem]{Claim}

\newtheorem{corollary}[theorem]{Corollary}
\newtheorem{lemma}[theorem]{Lemma}

\title{Towards the $k$-server conjecture:\\A unifying potential, pushing the frontier to the circle}
\author{Christian Coester\\CWI\\christian.coester@cwi.nl \and Elias Koutsoupias\\University of Oxford\\elias.koutsoupias@cs.ox.ac.uk}
\date{}

\begin{document}

\begin{titlepage}
	\maketitle
	\thispagestyle{empty}

\abstract{The $k$-server conjecture, first posed by Manasse, McGeoch and Sleator in 1988, states that a $k$-competitive deterministic algorithm for the $k$-server problem exists. It is conjectured that the work function algorithm (WFA) achieves this guarantee, a multi-purpose algorithm with applications to various online problems. This has been shown for several special cases: $k=2$, $(k+1)$-point metrics, $(k+2)$-point metrics, the line metric, weighted star metrics, and $k=3$ in the Manhattan plane.
	
	The known proofs of these results are based on potential functions tied to each particular special case, thus requiring six different potential functions for the six cases. We present a \emph{single} potential function proving $k$-competitiveness of WFA for all these cases. We also use this potential to show $k$-competitiveness of WFA on multiray spaces and for $k=3$ on trees. While the DoubleCoverage algorithm was known to be $k$-competitive for these latter cases, it has been open for WFA. Our potential captures a type of lazy adversary and thus shows that in all settled cases, the worst-case adversary is lazy. Chrobak and Larmore conjectured in 1992 that a potential capturing the lazy adversary would resolve the $k$-server conjecture.
	
	To our major surprise, this is not the case, as we show (using connections to the $k$-taxi problem) that our potential fails for three servers on the circle. Thus, our potential highlights laziness of the adversary as a fundamental property that is shared by all settled cases but violated in general. On the one hand, this weakens our confidence in the validity of the $k$-server conjecture. On the other hand, if the $k$-server conjecture holds, then we believe it can be proved by a variant of our potential.}

\end{titlepage}

\section{Introduction}

The $k$-server problem, introduced by Manasse, McGoech and Sleator \cite{ManasseMS88}, is one of the most fundamental problems in online optimization and contains other problems like paging or weighted paging as important special cases. It is defined as follows: $k$ servers are located in a metric space. One by one, points of the metric space are requested, and each request must be served upon arrival by moving one of the servers to the requested point. The problem is typically considered \emph{online}, where the choice of this server has to be made without knowledge of future requests. The goal is to minimize the total distance traveled by all servers.

When Manasse, McGeoch and Sleator \cite{ManasseMS88} introduced the $k$-server problem, they showed that on any metric space with $n\ge k+1$ points\footnote{On metric spaces with $n\le k$ points, the $k$-server problem is trivial.}, every deterministic online algorithm has competitive ratio at least $k$. They showed that this lower bound is tight when $k=2$ or $k=n-1$ by giving a $k$-competitive algorithm for these cases and boldly conjectured that a $k$-competitive online algorithm exists for the general case. This conjecture became known as the famous \emph{$k$-server conjecture} and has been a driving force in online optimization, making the $k$-server problem perhaps the most studied problem in the field. It has often been referred to as ``the holy grail of competitive analysis'', and many techniques developed for the $k$-server problem have later found applications to other problems.

Chrobak, Karloff, Payne, and Vishwanathan \cite{ChrobakKPV91} designed the elegant \emph{Double Coverage} algorithm to achieve the optimal competitive ratio of $k$ on the line metric. Shortly after, Chrobak and Larmore \cite{ChrobakL91} extended this algorithm to tree metrics, again matching the lower bound of $k$. The first algorithm for general metrics with a competitive ratio depending only on $k$ was found by Fiat, Rabani and Ravid \cite{FiatRR94}, achieving a competitive ratio exponential in $k$. Significant progress was made by Koutsoupias and Papadimitriou \cite{KoutsoupiasP95}, showing that a competitive ratio of $2k-1$ is achievable on general metric spaces.

While this reduces the gap between the upper and lower bound to a factor of $2$, it remains open to determine the exact competitive ratio. The lack of a proof of the $k$-server conjecture is even more puzzling given that the algorithm conjectured to achieve the competitive ratio of $k$ has been known for 30 years: The \emph{work function algorithm (WFA)}. It is this algorithm that achieves the aforementioned upper bound of $2k-1$ \cite{KoutsoupiasP95}. Its definition is generic\footnote{At any decision point, WFA chooses the action that would be best if the future were a mirror image of the past.}, with applications reaching far beyond the $k$-server problem. For instance, WFA achieves the optimal competitive ratio for metrical task systems \cite{BorodinLS92,BorodinE98}, the closely related \emph{generalized WFA} has been applied successfully to the weighted $k$-server problem \cite{BansalEK17}, the generalized $2$-server problem \cite{Sitters14} and layered graph traversal \cite{Burley96}, and work functions have also played a crucial role in recent breakthroughs for convex body chasing \cite{ArgueGGT20,Sellke20}. Given these connections, an exact understanding of the WFA for the $k$-server problem is likely to have a wider impact on online optimization in general.

WFA is known to achieve the tight competitive ratio of $k$ for the following special cases, which impose restrictions on the number of servers and/or the type of metric space:
\begin{itemize}
	\item $k=2$ \cite{ChrobakL92}
	\item $k=n-1$ (folklore; see e.g.~\cite{Koutsoupias09})
	\item $k=n-2$ \cite{KoutsoupiasP96,BartalK04}
	\item line metric \cite{BartalK04}
	\item weighted star metrics \cite{BartalK04}
	\item $k=3$ in the Manhattan plane \cite{BeinCL02}
\end{itemize}

While there has been a lack of progress on the $k$-server conjecture for about two decades, tremendous progress has been achieved for the \emph{randomized} $k$-server problem in recent years \cite{BansalBMN15,BCLLM18,Lee18}, leading to algorithms with polylogarithmic competitive ratios.

\subsection{Our contribution}

Our contribution consists of three parts.
\begin{enumerate}[(a)]
	\item The known proofs of the aforementioned six special cases where WFA is $k$-competitive all use a different potential function, and thus do not seem to point towards a potential function that can solve the $k$-server conjecture in the general case. We present a single potential function that proves the $k$-server conjecture for all these cases.
	\item Tree metrics are the only special case of the $k$-server problem where WFA is not known to be $k$-competitive but a different algorithm is (namely, the Double Coverage algorithm \cite{ChrobakL91}). In~\cite{BartalK04}, the question whether WFA is $k$-competitive on trees was raised as an intermediate step towards solving the $k$-server conjecture. In this direction, we use our potential function to show that WFA is $k$-competitive on multiray spaces (a type of tree metrics that generalizes the line and weighted star metrics) and for $k=3$ on general trees. Our proofs employ the quasi-convexity property of work functions in several new ways.
	\item Chrobak and Larmore~\cite{ChrobakL92} formulated three conjectures which say, essentially, that the ``adversary is lazy'' in the sense that at any time, the worst-case continuation of the request sequence begins with many requests to the $k$ offline server locations (forcing any sensible algorithm to converge to this configuration) before other points are requested. They verified their conjectures on tens of thousands of small metric spaces. In~\cite{BeinCL02}, a stronger statement was considered (ignoring the question what kind of work functions are ``reachable''), which fails in general but which they conjectured to be true on the circle metric. We reject all these conjectures by showing that for $k=3$, our potential captures exactly this lazy adversary (and a more restricted adversary for general $k$), but that it fails on the circle by giving an explicit request sequence as a counterexample. This highlights an important conceptual separation between all cases where $k$-competitiveness of WFA has been shown and the general case. We believe this property constitutes the main difficulty in resolving the $k$-server conjecture, and it suggests the circle as the main testing ground for further progress. Our method of constructing the counterexample is based on a connection with the $k$-taxi problem~\cite{CoesterK19}, which we use to generate phenomena of large metric spaces on a much smaller metric space.
\end{enumerate}

\subsection{Overview}
We provide various definitions and important lemmas in Section~\ref{sec:prelim}. In Section~\ref{sec:pot} we formally define our potential in two equivalent ways and show the basic way to use it to prove $k$-competitiveness. In Section~\ref{sec:interpretation}, we relate our potential to the lazy adversary potential that was defined implicitly by Chrobak and Larmore. We prove $k$-competitiveness on multiray spaces in Section~\ref{sec:multiray} and for $k=3$ on trees in Section~\ref{sec:trees}. The result on multi-ray spaces is our most involved proof, and implies the previously known $k$-competitiveness on the line and weighted stars as special cases. In Section~\ref{sec:nonlazy}, we provide a counter-example to our potential for $k=3$ on the circle, implying that the adversary is not lazy in this case, contrasting this case from all cases where WFA is known to be $k$-competitive. Additional proofs for previously known special cases using our potential are given in the appendix.
\section{Preliminaries}\label{sec:prelim}
\paragraph{Basic notation and abuse of notation.} We use $(M,d)$ to denote the metric space, where $d$ is the distance function. We denote by $n=|M|$ its size and by $\Delta=\max_{x,y}d(x,y)$ its diameter. For $x,y\in M$, we will often use the shorthand notation $xy:=d(x,y)$. A multiset $C\subseteq M$ of $k$ points is called a \emph{configuration}, representing the location of $k$ servers. We denote by $\Conf^k_M$ the set of all configurations. For two configurations, $X$ and $Y$, we denote by $d(X,Y)$ the value of their minimum matching. For notational convenience, we often use the empty space as a union operator on elements of $M$. For example, we often write $x_1x_2\dots x_i$ instead of $\{x_1,x_2,\dots,x_i\}$ when it is clear from the context that the set is meant. Similarly, given also a multiset $C$, we may write $Cx_1\dots x_i$ instead of $C\cup\{x_1,\dots,x_i\}$. For $x\in M$ and $i\in\N_0$, we write $x^i$ for the multiset containing $i$ copies of $x$.

For a set $S\subseteq M$ of points, let $\cl(S)$ for the sum of pairwise distances of the points in $S$.

\paragraph{The $k$-server problem.} An instance of the $k$-server problem is defined by a metric space $(M,d)$, an initial configuration $C_0\subseteq M$ of $k$ points and a sequence $r_1,r_2,\dots,r_T\in M$ of requests. A feasible solution is a sequence $C_1,C_2,\dots C_T$ of configurations such that $r_t\in C_t$ for all $t=1,\dots,T$. The cost of this solution is the sum $\sum_{t=1}^T d(C_{t-1},C_t)$.

\paragraph{The work function algorithm (WFA).} Given an instance of the $k$-server problem, the \emph{work function} $w_t$ at time $t$ is the function that maps any configuration $C$ to the minimal cost of serving the first $t$ requests and subsequently ending in configuration $C$. Formally,
\begin{align*}
	w_t(C):=\min_{\substack{C_1,\dots,C_t\\ \forall \tau\colon r_\tau\in C_\tau}}\sum_{\tau=1}^t d(C_{\tau-1},C_\tau) + d(C_t,C).
\end{align*}
The work function algorithm (WFA) selects $C_t\ni r_t$ so as to minimize $d(C_{t-1},C_t)+w_t(C_t)$, with ties broken arbitrarily.

\paragraph{Quasiconvexity.} A function $w\colon \Conf_M^k\to\R$ is called \emph{quasiconvex} if for any configurations $X$ and $Y$ there exists a bijection $\mu\colon X\to Y$ such that for any $A\subseteq X$,
\begin{align*}
	w(X)+w(Y)\ge w(A\cup\mu(X\setminus A)) + w(\mu(A)\cup (X\setminus A)).
\end{align*}
It was shown in \cite{KoutsoupiasP95} that if $w$ is quasiconvex, then $\mu$ can be chosen such that $\mu(x)=x$ for all $x\in X\cap Y$. More importantly, it was shown in \cite{KoutsoupiasP95} that any work function is quasiconvex.

\paragraph{Fundamentals about work functions.}
A function $w\colon \Conf_M^k\to\R$ is \emph{$1$-Lipschitz} if
\begin{align}
	w(X)-w(Y)\le d(X,Y)\label{eq:Lipschitz}
\end{align}
for all configurations $X$ and $Y$. The triangle inequality shows immediately that every work function is $1$-Lipschitz.

Let $\cQ^k_M$ be the set of functions $w\colon \Conf_M^k\to\R$ that are quasiconvex. Let $\cW^k_M\subseteq \cQ^k_M$ be the subset of functions that are additionally $1$-Lipschitz. We may drop $k$ and/or $M$ from the notation when they are clear from the context or immaterial. For $w\in\cW$ and configurations $X$ and $Y$, we say that $Y$ \emph{supports} $X$ if \eqref{eq:Lipschitz} holds with equality. Note that if $Y$ supports $X$ in $w_t$, then the cheapest way of serving the first $t$ requests and ending in configuration $X$ is equal to the cheapest way of serving the first $t$ requests and then first going to $Y$ and then to $X$. Thus, if $Y$ supports $X$, then there is no reason for an offline algorithm to be in configuration $X$ because it is at least as good to be in configuration $Y$ and delay the move from $Y$ to $X$ until later.

The \emph{support} of $w$, denoted $\supp(w)$, is the set of all configurations that are \emph{not} supported by any other configuration. Intuitively, $\supp(w_t)$ are the possible configurations where an optimal offline algorithm might be at time $t$. Clearly,
\begin{align*}
	w(X)=\min_{Y\in\supp(w)} w(Y)+d(X,Y)
\end{align*}
for any configuration $X$. In particular, any work function is fully specified by its support and the values it takes on support configurations.

For $r\in M$, let $\cW^k_M(r)\subseteq \cW^k_M$ be the subset of $1$-Lipschitz, quasiconvex functions with the property that every support configuration contains $r$. Again, we may drop $k$ and/or $M$ from the notation. Note that the work function $w_t$ at time $t$ is in $\cW(r_t)$.

There exists a simple update rule to compute the new work function when an additional request is issued. For $w\in \cW_M$ and $r\in M$, the updated work function $w\land r\in\cW(r)$ is defined by
\begin{align*}
	w\land r(C)= \min_{X\ni r} w(X)+d(X,C).
\end{align*}
It is easy to see that $w_t=w_{t-1}\land r_t$. A basic observation is that if $r_t\in C$, then $w_{t-1}(C)=w_t(C)$. Another basic property is that $w_t(C)\ge w_{t-1}(C)$.

\subsection{Extended cost, minimizers and duality}

The following lemma was proved by Chrobak and Larmore \cite{ChrobakL92} (see also \cite{Koutsoupias09}):
\begin{lemma}[Extended cost lemma]
	If for every $k$-server instance on a metric space $M$ it holds that
	\begin{align*}
		\sum_{t=1}^T\max_X\left[w_t(X)-w_{t-1}(X)\right]\le (\rho+1)\cdot\min_{X}w_T(X)+c_M
	\end{align*}
	for some constant $c_M$ independent of the request sequence, then WFA is $\rho$-competitive on $M$.
\end{lemma}

The power of this lemma is that it reduces the task of proving competitiveness of WFA to a property of work functions. In particular, we do not need to keep track of the actual configurations of the online and offline algorithm. The quantity $\max_X\left[w_t(X)-w_{t-1}(X)\right]$ is also called the \emph{extended cost} of the $t$th request, and the proof of the lemma is based on the fact that the total extended cost over all requests is an upper bound on the sum of WFA's cost and the optimal offline cost.

For a work function $w\in\cW^k_M$ and a point $y\in M$, we call a configuration $X\in\arg\min w(X)-d(y^k,X)$ a \emph{minimizer of $w$ with respect to $y$}. There is a direct connection between minimizers and the configurations $X$ maximizing the extended cost. %In particular, if $X$ is a minimizer of $w_{t-1}$ with respect to $r_t$, then the extended cost of the $t$th requests is equal to $w_t(X)-w_{t-1}(X)$.
This connection is captured by the duality lemma, which was first proved in~\cite{KoutsoupiasP95}. We give a slightly stronger version of the duality lemma by stating it as an equivalence rather than an implication.

\begin{lemma}[Duality lemma]
	Let $w\in\cW_M$ and $r\in M$. Define $w'=w\land r$. Then $A\in \arg\min_X w(X)-d(r^k,X)$ if and
	only if the following two conditions hold:
	\begin{align}
		A\in \arg\max_X w'(X)-w(X) \\
		A\in \arg\min_X w'(X)-d(r^k,X)
	\end{align}
\end{lemma}     
\begin{proof}
	The ``only if'' direction is the duality lemma of~\cite{KoutsoupiasP95}, where it was shown that if $A\in \arg\min_X w(X)-d(r^k,X)$ then for every configuration $B$
	\begin{align}
		w'(A)+w(B) &\geq w(A)+w'(B), & \text{and} \\
		w'(B)-d(r^k,B) &\geq w'(A)-d(r^k,A). 
	\end{align}
	By summing these two constraints we get $w(B)-d(r^k,B)\geq w(A)-d(r^k,A)$,
	which shows the other direction.
\end{proof}

It is interesting that the proof of the duality lemma does not use the fact that
$d$ is a distance, i.e., it satisfies the triangle inequality.

\subsection{Additional properties of work functions}

In this section, we provide additional properties of work functions that follow from the quasiconvexity property. We will use these properties to prove $k$-competitiveness on multiray spaces and for $k=3$ on trees.

The notion of quasiconvex or quasiconcave functions appears in many different areas and was
discovered independently a few times. As a result, they appear with different
terminology in literature. For example, in the early 1980s Celso and
Crawford~\cite{kelso1982job} defined a related notion as a sufficient condition to
the existence of Walrasian Equilibria and called a similar notion gross substitute
functions\footnote{Gross substitute functions are real functions defined for all
	subsets of a ground set $V$, whose restriction to subsets of each size $k$ are
	quasiconvex.}; in 1990, Dress and Wenzel~\cite{Dress1990} related them to a variant
of the greedy
algorithm and called them valuated matroids; Koutsoupias and
Papadimitriou~\cite{KoutsoupiasP95} defined them in the context of online algorithms for the
$k$-server problem and called them quasiconvex. They have also played a central
role in discrete optimization~\cite{Murota2003c}.

\begin{lemma}\label{lem:quasiMin}
	Let $w\in\cQ$. Let $X\in\arg\min w(X)$, and let $x\in X$. Then there exists $Y\in \arg\min_{Y\not\ni x}w(Y)$ such that $X-x\subset Y$.
\end{lemma}
\begin{proof}
	Let $Y$ be chosen such that $X\cap Y$ is maximal under inclusion. Suppose towards a contradiction that there exists $x'\in (X-x)\setminus Y$. By quasiconvexity, there exists $y'\in Y\setminus X$ such that $w(X)+w(Y)\ge w(X-x'+y')+w(Y-y'+x')$. By choice of $X$, we have $w(X-x'+y')\ge w(X)$. Combining these last two inequalities, we get $w(Y)\ge w(Y-y'+x')$. But $Y-y'+x'\not\ni x$ and $X\cap Y\subsetneq X\cap (Y-y'+x')$, so this contradicts the choice of $Y$.
\end{proof}

\begin{lemma}\label{lem:quasiSub}
	Let $w\in \cQ_M^k$. Let $X\in\arg\min w(X)$, and let $A\subset M$ be a (multi)set of cardinality $|A|<k$. Then there exists $Y\in\arg\min_{Y\supset A} w(Y)$ such that $Y- A\subseteq X- A$.
\end{lemma}
\begin{proof}
	Let $Y$ be chosen such that $(Y-A)\setminus (X-A)$ is minimal under inclusion and suppose towards a contradiction that there exists $y\in (Y-A)\setminus (X-A)$. By quasiconvexity, there exists $x\in X \setminus Y$ such that $w(X)+w(Y)\ge w(X-x+y)+w(Y-y+x)$. By choice of $X$, we have $w(X-x+y)\ge w(X)$. Combining these inequalities, we get $w(Y)\ge w(Y-y+x)$. But this contradicts the choice of $Y$ since we would rather have chosen $Y-y+x$.
\end{proof}

The next lemma shows that the greedy algorithm can be used to find a configuration
that minimizes the value of a quasiconvex function (as it was first shown
in~\cite{Dress1990}). The greedy algorithm starts with an arbitrary configuration
and processes its elements in order, replacing each element with the currently best
choice. That is, when the current configuration is $X$ and the algorithm processes
element $x\in X$, it replaces it with $x^*\in \arg\min_y w(X-x+y)$.

\begin{lemma}\label{lem:quasiGreedy}
	Let $w\in\cQ$, $A\subset M$ and $Y\in\arg\min_{Y\supseteq A} w(Y)$. Then there exists $X\in\arg\min_X w(X)$ with $Y-A\subset X-A$.
\end{lemma}
\begin{proof}
	Similar to the proof of \pref{lem:quasiSub}.
\end{proof}

\begin{lemma}\label{lem:resolveMonotone}
	Let $w\in\cW^k_M(r)$, let $X\subseteq M$ be a $k$-point multiset and $x,y\in X$. If $X$ resolves\footnote{When $w(X)=w(X-x+y)+xy$, we say that $X$ ``resolves from $x$ to $y$''. If $y=r$ is the last request, we simply say that $X$ resolves from $x$. } from $x$ in $w$, then also $X-y+x$ resolves from $x$ in $w$.%\footnote{\ccnote{Need to define terminology %``resolves from''}}
\end{lemma}
\begin{proof}
	Suppose that instead, $X-y+x$ resolves from some $z\in X-y-x$. Consider the $(k-3)$-point multiset $C:= X-y-x-z$.  Then
	\begin{align*}
	w(X) + w(X-y+x)&=w(Cxyz) + w(Cx^2z)\\
	&= w(Cyzr) + w(Cx^2r) +rx+ rz\\
	&\ge w(Cxyr) + w(Cxzr) +rx+rz\\
	&\ge w(Cxyz) + w(Cx^2 z),
	\end{align*}
	where the first inequality is by quasiconvexity and the last by $1$-Lipschitzness of $w$. Since the second and the last expression are the same, we have equality in all steps. But then the last step shows that $Cx^2z$ resolves from $x$. Since $Cx^2z=X-y+x$, the lemma follows.
\end{proof}

\section{The potential}\label{sec:pot}
We provide two different, but equivalent definitions of our potential function. The first formulation views the potential through the lens of the \emph{$m$-evader problem}, which is equivalent to the $k$-server problem when $m=n-k$. Thereafter, we will give a more compact and equivalent formulation of the same potential in the $k$-server view based on extending the metric space by adding antipodal points.

\subsection{The evader potential}
The $m$-evader problem is defined similarly to the $k$-server problem, but instead of $k$ servers there are $m$ evaders in the metric space, which must occupy $m$ different points at all times. When a point $r$ is requested, rather than moving a server towards $r$, an evader that might be located at $r$ has to move to a different point. The equivalence between the $k$-server problem and the $(n-k)$-evader problem follows by identifying a server configuration $C$ with the evader configuration $M\setminus C$.\footnote{This identification requires the server configuration to be a set rather than a multiset. This is no restriction on the power of $k$-server algorithms (online or offline).} Given a $k$-server work function $w$, we denote by $\hat w$ the corresponding evader work function, defined by $\hat w(C):=w(M\setminus C)$.

In the evader view, the potential $\hat\Phi$ is defined as follows. Let $y=(y_1,\dots,y_n)$ be a permutation of the points of the metric space $M$. Let
\begin{align}
	\hat{\Phi}_y(\hw)&:=\cl(y_1\dots y_{n-k-1})+\sum_{i=n-k}^n \min_{\substack{C\subseteq\{y_1,\dots,y_{i}\}\\|C|=n-k}}\left(\hw(C)+d(C,y_i^{n-k})\right)\nonumber\\
	\hat{\Phi}(\hw)&:=\min_y\hat{\Phi}_y(\hw).\label{eq:permPot}
\end{align}

\begin{theorem}\label{thm:evPot}
	Let $(M,d)$ be an $n$-point metric space. If for every $r\in M$ and every work function $w\in\cW_M^k(r)$ it holds that $\hPhi(\hw)=\hPhi_y(\hw)$ for a permutation $y$ of $M$ with $y_n=r$, then WFA is $k$-competitive on $M$.
\end{theorem}
\begin{proof}
	Consider a $k$-server instance on $M$ with a request sequence $r_1,\dots,r_T$ and associated sequence of work functions $w_0,\dots,w_T$. We first show that at each time $t$, the change in potential is an upper bound on the extended cost.
	
	By the premise of the lemma, $\hPhi(\hw_t)=\hPhi_y(\hw_t)$ for some $y$ with $y_n=r_t$. Thus,
	\begin{align*}
		\hPhi(\hw_t)-\hPhi(\hw_{t-1})&\ge \hPhi_y(\hw_t)-\hPhi_y(\hw_{t-1})\\
		&\ge \min_{\substack{C\subseteq M\\|C|=n-k}}\left(\hw_t(C)+d(C,r_t^{n-k})\right) - \min_{\substack{C\subseteq M\\|C|=n-k}}\left(\hw_{t-1}(C)+d(C,r_t^{n-k})\right)\\
		&= \min_{\substack{X\subseteq M\\|X|=k}}\left(w_t(X)-d(X,r_t^k)\right) - \min_{\substack{X\subseteq M\\|X|=k}}\left(w_{t-1}(X)-d(X,r_t^k)\right)\\
		&=\max_{X}w_t(X)-w_{t-1}(X),
	\end{align*}
	where the first inequality uses $\hPhi(\hw_{t-1})\le \hPhi_y(\hw_{t-1})$, the second inequality uses $y_n=r_t$ and the fact that $\hw_{t-1}(C)\le \hw_t(C)$ for each $C$, the first equation translates evader work functions to server work functions and uses $d(C,r_t^{n-k})=d(M,r_t^n)-d(M\setminus C,r_t^k)$, and the second equation is due to the duality lemma, which says that the same $X$ can be chosen in both minima and the maximum. So indeed, the change in potential upper bounds the extended cost.
	
	Now, we can bound the total extended cost by
	\begin{align*}
		\sum_{t=1}^T\max_X\left[w_t(X)-w_{t-1}(X)\right]&\le \hPhi(\hw_T)\\
		&\le (k+1)\cdot\min_{X}w_T(X)+c_M,
	\end{align*}
	where the last inequality is due to the fact that $\hPhi(\hw_T)$ is a sum of distances (which are absorbed by the constant $c_M$) and $k+1$ work function values, each of which differs from $\min_{X}w_T(X)$ by at most $k$ times the diameter of $M$ due to $1$-Lipschitzness of $w_T$ (and the diameters are also absorbed by $c_M$). The theorem now follows from the extended cost lemma.
\end{proof}

\subsection{\texorpdfstring{The $k$-server potential}{The k-server potential}}
We now derive an equivalent but simpler expression for the aforementioned potential. To formulate it, we need the notion of antipodal points.

Let $\Delta$ be the diameter of $M$. A point $\bar{p}\in M$ is called the \emph{antipode} of another point $p\in M$ if for each $x\in M$, $px+x\bar p=p\bar{p}=\Delta$. On some metric spaces such as the circle, every point has an antipode. As mentioned in~\cite{Koutsoupias99}, every metric space can be extended so that every point has an antipode: To achieve this, we add to $M$ another copy of the same points, $\bar M=\{\bar p\colon p\in M\}$, and define distances by $\bar p\bar q=pq$ and $\bar pq=2\Delta-pq$ for $p,q\in M$. It is easy to check that $M\cup\bar M$ is still a metric space (of diameter $2\Delta$) where $\bar p$ and $p$ are antipodes of each other.

Consider a metric space $M$ where every point has an antipode. Let $x_1,\dots,x_k\in M$. We define the $k$-server potential $\Phi$ via
\begin{align}
	\Phi_{x_1,\dots,x_k}(w)&:=\sum_{i=0}^k w(\bar{x}_i^i x_{i+1}\dots x_k)\nonumber\\
	\Phi(w) &:= \min_{x_1,\dots,x_k}\Phi_{x_1,\dots,x_k}(w).\label{eq:antiPot}
\end{align}

The following lemma states that the two potential functions differ by a fixed constant depending on $M$ and are therefore equivalent.
\begin{lemma}
	Let $M$ be a pseudo-metric space of diameter $\Delta$ where every point has an antipode and there are $k$ copies of each point.\footnote{It is only a pseudo-metric because the distance between two copies of the same point is $0$. We use the assumption of several copies of the same point because the definition of $\Phi_{x_1,\dots,x_k}$ allows points to repeat, whereas $\Phi_y$ requires $y$ to be a permutation.} For any work function $w\in \cW_M^k$ and any permutation $y=(y_1,\dots,y_n)$ of $M$,
	\begin{align*}
		\Phi_{y_{n-k+1}\dots y_n}(w)=\hPhi_y(\hw)-\cl(M)+\frac{k(k+1)}{2}\Delta.
	\end{align*}
\end{lemma}
\begin{proof}
	Subtracting $\cl(M)$ from the evader potential and using server work functions instead of evader work functions, we have
	\begin{align*}
		\hPhi_y(\hw)-\cl(M) &= \sum_{i=n-k}^n \min_{\substack{C\supseteq\{y_{i+1},\dots,y_{n}\}\\|C|=k}}\left(w(C)-\sum_{p\in C\cap\{y_1,\dots,y_i\}}py_i\right).
	\end{align*}
	Notice that the minimum in the summand for $i$ is achieved when $C\setminus\{y_{i+1},\dots,y_n\}$ consists of $k-n+i$ copies of the antipodal point $\bar y_i$. Thus, the expression is equal to
	\begin{align*}
		\sum_{i=n-k}^n \left(w(\bar{y_i}^{k-n+i}y_{i+1}\dots y_n)-(k-n+i)\Delta\right) &= \sum_{i=n-k}^n w(\bar{y_i}^{k-n+i}y_{i+1}\dots y_n)-\frac{k(k+1)}{2}\Delta\\
		&=\Phi_{y_{n-k+1}\dots y_n}(w)-\frac{k(k+1)}{2}\Delta.\qedhere
	\end{align*}
\end{proof}

\begin{corollary}\label{cor:serverPot}
	Let $(M,d)$ be a metric space where every point has an antipode. If for every $r\in M$ and every work function $w\in\cW_M^k(r)$ it holds that $\Phi(w)=\Phi_{x_1\dots x_k}(w)$ for some $x_1,\dots x_k\in M$ with $x_k=r$, then WFA is $k$-competitive on $M$.
\end{corollary}

\section{Interpretation as a lazy adversary potential}\label{sec:interpretation}
\subsection{The implicitly defined potential by Chrobak and Larmore}
Chrobak and Larmore \cite{ChrobakL92} gave an \emph{implicit} definition of a potential that they conjectured to prove the $k$-server conjecture. This potential captures exactly a type of lazy adversary. To give a precise definition, we first need some additional notation.

For $r\in M$ and a work function $w\in \cW$, denote by $\nabla(w,r):=\max_{A}(w\land r)(A)-w(A)$ the extended cost of request $r$ on $w$. For a request sequence $\rho=(r_1,\dots,r_T)\in M^*$, let
\begin{align*}
\nabla(w,\rho):=\sum_{t=1}^T\nabla(w_{t-1},r_t)
\end{align*}
be the total extended cost, where $w_t=w\land r_1\land r_2\land \dots \land r_t$ is the updated work function after the first $t$ requests. The potential conjectured by Chrobak and Larmore is given by
\begin{align*}
\tPhi(w)&:=\min_X\tPhi_X(w)\qquad\qquad\qquad
\intertext{where the maximum is taken over configurations $X$ and}
\tPhi_{X}(w)&:=-\cl(X) + (k+1)w(X)-\sup_{\rho\in X^*}\nabla(w,\rho).
\end{align*}
Because of the term $\sup_{\rho\in X^*}\nabla(w,\rho)$, this potential captures exactly the worst-case extended cost when the future request sequence consists only of points in $X$, until the work function is a cone\footnote{A work function is a \emph{cone} if its support contains only a single configuration.} with support $\{X\}$. An adversary constructing such a request sequence can be thought of as ``lazy'' because it wants to force the online algorithm to the offline configuration $X$ before it requests different points. The additional term $\cl(X)$ is needed because of extended cost being incurred when passing from one cone to a different cone. The definition of $\tPhi$ is only implicit because of the supremum over request sequences $\rho\in X^*$. It was conjectured in \cite{ChrobakL92} that $\tPhi(w\land r)-\tPhi(w)\ge \nabla(w,r)$ for any (reachable) work function $w$ and request $r$. This would imply the $k$-server conjecture similarly to the proof of Theorem~\ref{thm:evPot}. They also conjectured that $\tPhi_X(w\land r)$ is minimized for a configuration $X$ containing $r$, and more specifically that it is minimized by a configuration $X\in\supp(w\land r)$. This would imply the previous conjectures. We show that for $k=3$, the potential $\tPhi$ matches our potential $\Phi$. For $k\ge 4$, our potential captures a more restricted type of lazy adversary. As we will show in Section~\ref{sec:nonlazy} that our potential fails to bound the extended cost for $k=3$ on the circle, this disproves the conjectures from \cite{ChrobakL92} and yields the surprising insight that the worst-case adversary on the circle is not lazy -- unlike the adversary for all cases where WFA is known to be $k$-competitive.

\subsection{Relationship to our potential}
Our next lemma shows that our potential $\Phi$ captures a more restricted adversarial strategy, where the configuration $X$ is ordered as $x_1,\dots,x_k$ and the next request in $\rho$ is always to the point $x_i$ with $i$ maximal that leads to a change of the work function. We will show later that for $k=3$, this imposes no additional restriction.

For fixed $x_1,\dots,x_k\in M$ and a work function $w\in \cW_M^k$, define a request sequence $r_1,r_2,\dots,r_T$ as follows. Let $w_t=w\land r_1\land r_2\land \dots \land r_t$ be the updated work function after the first $t$ requests. We define $r_t=x_i$ for $i$ maximal such that $w_{t-1}\land x_i\ne w_{t-1}$; if no such $i$ exists, the request sequence ends, $T=t-1$, and $w_{T}$ is a cone with support $\{\{x_1,\dots,x_k\}\}$. 

\begin{lemma}\label{lem:permIntuition}
	\begin{align*}
	\Phi_{x_1,\dots,x_k}(w)=\frac{k(k+1)}{2}\Delta-\cl(x_1,\dots,x_k) + (k+1)w(x_1\dots x_k)-\sum_{t=1}^T\nabla(w_{t-1},r_t).
	\end{align*}
\end{lemma}
\begin{proof}
	It suffices to show
	\begin{align}
	\Phi_{x_1,\dots,x_k}(w_t)&=\Phi_{x_1,\dots,x_k}(w_{t-1})+\nabla(w_{t-1},r_t)\label{eq:step}\\
	\Phi_{x_1,\dots,x_k}(w_T)&=(k+1)w(x_1\dots x_k)+\frac{k(k+1)}{2}\Delta-\sum_{1\le i<j\le k}x_ix_j.\label{eq:cone}
	\end{align}
	For equation \pref{eq:cone}, we have
	\begin{align*}
	\Phi_{x_1,\dots,x_k}(w_T) &= \sum_{j=0}^k w_T(\bar{x}_j^j x_{j+1}\dots x_k)\\
	&= (k+1)w_T(x_1,\dots,x_k)+\sum_{1\le i\le j\le k}x_i\bar{x}_j\\
	&= (k+1)w(x_1,\dots,x_k)+\sum_{1\le i\le j\le k}(\Delta-x_ix_j)\\
	&= (k+1)w(x_1,\dots,x_k)+\frac{k(k+1)}{2}\Delta-\sum_{1\le i<j\le k}x_ix_j.
	\end{align*}
	We now show equation \pref{eq:step}. Let $i$ be such that $r_t=x_i$. Then,
	\begin{align}
	\Phi_{x_1,\dots,x_k}(w_t)&= \sum_{j=0}^k (w_{t-1}\land x_i)(\bar{x}_j^j x_{j+1}\dots x_k)\nonumber\\
	&= \sum_{j=0}^{i-1} w_{t-1}(\bar{x}_j^j x_{j+1}\dots x_k) + \sum_{j=i}^k (w_{t-1}\land x_i)(\bar{x}_j^j x_{j+1}\dots x_k).\label{eq:splitSum}
	\end{align}
	By maximality of $i$, $x_{j+1}\dots x_k$ is contained in every support configuration of $w_{t-1}$. Thus, $\bar{x}_i^ix_{i+1}\dots x_k$ is a minimizer of $w_{t-1}$ with respect to $x_i$ and hence
	\begin{align}
	(w_{t-1}\land x_i)(\bar{x}_i^i x_{i+1}\dots x_k) = w_{t-1}(\bar{x}_i^i x_{i+1}\dots x_k) + \nabla(w_{t-1},r_t)\label{eq:i=j}
	\end{align}
	by the duality lemma.
	
	We claim that
	\begin{align}
	(w_{t-1}\land x_i)(\bar{x}_j^j x_{j+1}\dots x_k)& =w_{t-1}(\bar{x}_j^j x_{j+1}\dots x_k)\qquad\forall j=i+1,\dots, k.\label{eq:bigJ}
	\end{align}
	Assuming this is true, we obtain \pref{eq:step} by substituting \pref{eq:i=j} and \pref{eq:bigJ} into \pref{eq:splitSum}.
	
	It remains to show \pref{eq:bigJ}. Since $w_{t-1}\land x_i\ge w_{t-1}$, the direction ``$\ge$'' is immediate. For the other direction, since $x_jx_{j+1}\dots x_k$ is contained in every support configuration of $w_{t-1}$, we get
	\begin{align*}
	w_{t-1}(\bar{x}_j^j x_{j+1}\dots x_k)&= w_{t-1}(\bar{x}_j^{j-1}x_j x_{j+1}\dots x_k) + \bar{x}_jx_j\\
	&\ge w_{t-1}(\bar{x}_j^{j-1}x_i x_{j+1}\dots x_k) - x_ix_j+\bar{x}_jx_j\\
	&= (w_{t-1}\land x_i)(\bar{x}_j^{j-1}x_i x_{j+1}\dots x_k) + x_i\bar{x}_j\\
	&\ge (w_{t-1}\land x_i)(\bar{x}_j^{j}x_{j+1}\dots x_k).\qedhere
	\end{align*}
\end{proof}

\begin{lemma}\label{lem:push3}
	Let $X\subset M$ with $|X|=3$ and $r\in X$ be fixed and let $w\in\cW_M^3(r)$. For a bijection $\pi\colon \{1,\dots,3\}\to X$, write $\Phi_{\pi}:=\Phi_{\pi(1)\pi(2)\pi(3)}$. Then
	\begin{align*}
	\min_{\pi\colon \pi(3)=r}\Phi_\pi(w) = \min_{\pi}\Phi_\pi(w).
	\end{align*}
\end{lemma}
\begin{proof}
	Let $\pi$ be a minimizer of the right hand side. If $\pi(k)=r$, we are done. The case $\pi(k-1)=r$ is also easy, using the fact that $r$ is contained in every support configuration. The remaining case $\pi(k-2)=r$ is non-trivial. Let $y:=\pi(k-1)$ and $z:=\pi(k)$. We will construct a permutation $\pi'$ with $\pi'(3)=r$ and $\Phi_\pi(w)\ge\Phi_{\pi'}(w)$. This will only affect the last three terms in the sum of the definition of $\Phi$,
	\begin{align*}
	w(\bar{r}^{k-2}yz)+w(\bar{y}^{k-1}z)+w(\bar{z}^k).
	\end{align*}
	
	If $w(\bar{y}^{k-1}z)=w(\bar{y}^{k-2}rz)+\bar yr$, then
	\begin{align*}
	w(\bar{r}^{k-2}yz)+w(\bar{y}^{k-1}z)+w(\bar{z}^k)&=w(\bar{r}^{k-2}yz)+w(\bar{y}^{k-2}rz)+w(\bar{z}^k)+\bar yr\\
	&\ge w(\bar{y}^{k-2}rz)+w(\bar{r}^{k-1}z)+w(\bar{z}^k)
	\end{align*}
	where the inequality uses $\bar yr=y\bar r$. This corresponds to a permutation with $r$ in the next-to-last position, and it is easy to push it from there to the last position.
	
	So we can assume $w(\bar{y}^{k-1}z)=w(\bar{y}^{k-1}r)+zr$. Thus
	\begin{align}
	w(\bar{r}^{k-2}yz)+w(\bar{y}^{k-1}z)+w(\bar{z}^k)&=w(\bar{r}^{k-2}yz)+w(\bar{y}^{k-1}r)+w(\bar{z}^{k-1}r)+zr+\bar z r\nonumber\\
	&=w(\bar{r}^{k-2}yz)+w(\bar{y}^{k-1}r)+w(\bar{z}^{k-1}r)+\Delta.\label{eq:push31}
	\end{align}
	
	In the last expression, $y$ and $z$ are symmetric, so we can assume
	\begin{align}
	w(\bar{r}^{k-2}yz)=w(\bar{r}^{k-2}rz)+yr.\label{eq:push32}
	\end{align}
	By quasi-convexity and Lipschitzness of the work function (and $\bar y\bar r = yr$, $\bar r r=\Delta$),
	\begin{align}
	w(\bar{y}^{k-1}r)+w(\bar{r}^{k-2}rz)&\ge w(\bar{y}^{k-2}zr)+w(\bar{r}^{k-2}\bar yr)\nonumber\\
	&\ge w(\bar{y}^{k-2}zr)+ w(\bar{r}^{k})-yr-\Delta\label{eq:push33}
	\end{align}
	Combining \pref{eq:push31}, \pref{eq:push32} and \pref{eq:push33}, we get
	\begin{align*}
	w(\bar{r}^{k-2}yz)+w(\bar{y}^{k-1}z)+w(\bar{z}^k) \ge w(\bar{y}^{k-2}zr)+w(\bar{z}^{k-1}r)+w(\bar{r}^{k}),
	\end{align*}
	corresponding to the permutation $(\pi(1),\pi(2),\pi(3))=(y,z,r)$.
\end{proof}

We remark (without proof) that the above lemma fails for $k=4$.

By the following corollary, for $k=3$ it holds that $\Phi$ is an explicit expression for the implicit potential of \cite{ChrobakL92}.

\begin{corollary}
	For $k=3$,
	\begin{align}
	\Phi(w)=6\Delta + \min_{\substack{X\colon |X|=3\\ r_1,\dots,r_T\in X}}\left[4w(X)-\cl(X)-\sum_{t=1}^T\nabla(w\land r_1\dots r_{t-1},r_t)\right].\label{eq:permIntuition3}
	\end{align}
\end{corollary}
\begin{proof}
	The direction ``$\ge$'' follows from \pref{lem:permIntuition}. For the direction ``$\le$'', select $X$ and $r_1,\dots,r_T$ to minimize the right hand side. Write $w_t=w\land r_1\dots r_{t}$. By minimality of the right hand side, $w_T$ is a cone at $X$. Let $\Phi_X=\min\Phi_{x_1x_2x_3}$, where the minimum is taken over permutations $x_1,x_2,x_3$ of $X$. By \pref{lem:push3}, we have $\Phi_X(w_t)=\Phi_{xyr_t}(w_t)$ for some $x,y\in X$. Thus,
	\begin{align*}
	\Phi_X(w_t)-\Phi_X(w_{t-1})&\ge \Phi_{xyr_t}(w_t)-\Phi_{xyr_t}(w_{t-1})\\
	&= w_t(\bar r_t^3)-w_{t-1}(\bar r_t^3)\\
	&=\nabla(w_{t-1},r_t).
	\end{align*}
	Hence,
	\begin{align*}
	\Phi(w)\le \Phi_X(w)&=\Phi_X(w_T)-\sum_{t=1}^T\left[\Phi_X(w_{t})-\Phi_X(w_{t-1})\right]\\
	&\le \Phi_X(w_T)-\sum_{t=1}^T\nabla(w_{t-1},r_t),
	\end{align*}
	which is equal to the right hand side of \pref{eq:permIntuition3} due to \pref{lem:permIntuition} and since $w_T$ is a cone at $X$.
\end{proof}

\section{Multi-ray spaces}\label{sec:multiray}
A multi-ray space is a tree of depth $1$ whose edges have infinite length and where requests can appear at arbitrary locations along the edges. We call these edges rays.

We will show in this section that WFA is $k$-competitive on multiray spaces. Note that a multiray space with only $2$ rays is equal to the line metric. A subset of a multi-ray space containing only one point from each ray is a weighted star. Our proof therefore recovers the known proofs that WFA is $k$-competitive on the line and on weighted stars as special cases.

We denote by $c$ the center/root of the multi-ray space, i.e., the origin of the rays. We can assume that every ray has finite length by considering only a sufficiently long part that all requests fall into. We call the endpoint of a ray that is not the center a \emph{leaf}. Denote by $\cL$ the set of leaves. For $w\in \cW^k$, define $m_w(X):=w(X)-d(c^k,X)$. Note that $m_w$ is also quasiconvex. As we use the server definition of the potential, we augment the multi-ray space by adding antipodes as discussed earlier. In the definition \pref{eq:permPot}, we require the points $x_1,\dots,x_k$ to be chosen from the original metric space $M$. This corresponds to requiring the permutation in the evader potential to end with $k$ points from the original metric space, which does not affect the proof of Theorem~\ref{thm:evPot}.

The proof that WFA is $k$-competitive on multi-ray spaces proceeds along the following three main steps:
\begin{enumerate}
	\item First we establish some properties of $\Phi_{x_1\dots x_k}$ when $x_i=\ell_i$ are leaves. In particular, we express $\Phi_{\ell_1\dots\ell_k}$ in terms of $m_w$, and show that $\ell_1,\dots,\ell_k$ can be permuted under certain conditions.
	\item We then show by induction on $k$ that $\Phi_{x_1\dots x_k}(w)$ is indeed minimized when $x_1,\dots, x_k$ are leaves and $\min_X m_w(X)=m_w(x_1\dots x_k)$.
	\item Finally, we show that $\Phi_{x_1\dots x_k}(w)$ is also minimized for some $x_1,\dots,x_k$ where only $x_1,\dots,x_{k-1}$ are leaves whereas $x_k=r$ is the last request.
\end{enumerate}

\subsection*{Step 1: Properties of $\Phi_{x_1\dots x_k}$ when $x_i$ are leaves}
\begin{lemma}
	Let $w\in \cW^k$. There exist leaves $\ell_1,\dots,\ell_k$ such that $\min_X m_w(X)=m_w(\ell_1\dots\ell_k)$.
\end{lemma}
\begin{proof}
	Follows from the fact that since $w$ is $1$-Lipschitz, $m_w(X)$ cannot increase when a point in $X$ moves away from $c$ towards a leaf.
\end{proof}

\begin{lemma}\label{lem:potTermLeaves}
	For any $w\in\cW^k$, a leaf $\ell\in\cL$ and $x_{i+1},\dots,x_k\in M$,
	\begin{align*}
		w(\bl^ix_{i+1}\dots x_k)&=\min_{\substack{X\supseteq x_{i+1}\dots x_k\colon\\X- x_{i+1}\dots x_k\subseteq\cL-\ell}}m_w(X) + i(\Delta-c\ell) + \sum_{j={i+1}}^k cx_j.
	\end{align*}
\end{lemma}
\begin{proof}
	\begin{align*}
		w(\bl^ix_{i+1}\dots x_k)&=\min_{X\supseteq x_{i+1}\dots x_k}w(X) + d(X- x_{i+1}\dots x_k,\bl^i)\\
		&=\min_{X\supseteq x_{i+1}\dots x_k}w(X) + i\Delta - d(X- x_{i+1}\dots x_k,\ell^i)\\
	\end{align*}
	We claim that the minimum is achieved by some $X$ with $X- x_{i+1}\dots x_k\subseteq\cL-\ell$: Indeed, if there is some $x\in X- x_{i+1}\dots x_k$ that is not in $\cL-\ell$, sliding $x$ away from $\ell$ along a path to some other leaf increases $d(X- x_{i+1}\dots x_k,\ell^i)$ by the distance moved, and it increases $w(X)$ by at most this distance, so the whole term cannot increase.
	
	This also means that every distance in $d(X- x_{i+1}\dots x_k,\ell^i)$ goes across the center, i.e.,
	\begin{align*}
		d(X- x_{i+1}\dots x_k,\ell^i) = i\cdot c\ell + d(X- x_{i+1}\dots x_k,c^i).
	\end{align*}
	The lemma now follows by definition of $m_w$.
\end{proof}

As a consequence of \pref{lem:potTermLeaves}, we obtain the following expression for $\Phi_{\ell_1\dots\ell_k}$ whenever $\ell_1,\dots,\ell_k$ are leaves:
\begin{align}
	\Phi_{\ell_1\dots\ell_k}(w) &= \frac{k(k+1)}{2}\Delta+\sum_{i=0}^k\min_{\substack{X_i\supseteq \ell_{i+1}\dots \ell_k\colon\\X_i- \ell_{i+1}\dots \ell_k\subseteq\cL-\ell_i}}m_w(X_i).\label{eq:potLeaves}
\end{align}

The following symmetry and monotonicity properties allow us to reorder $\ell_1,\dots,\ell_k$ under certain circumstances.

%\ccnote{Note that when every ray has the same length and $X$ consists only of leaves, then $m_w(X)$ and $w(X)$ are actually equal up to a constant ($k$ times the ray length). So we could replace $m_w$ by $w$ in \pref{eq:potLeaves}.}

\begin{lemma}[Symmetry and Monotonicity Lemma]\label{lem:symmetry}
	Let $w\in \cW^k$ and let $\ell_1,\dots,\ell_k$ be leaves such that $\min_X m_w(X)=m_w(\ell_1\dots\ell_k)$. The following properties hold:
	\begin{description}
		\item[Symmetry:] $\Phi_{\ell_1\dots \ell_k}(w)$ is constant under permutation of $\ell_1,\dots,\ell_k$.
		
		\item[Monotonicity:] For any leaf $\ell$, $\Phi_{\ell_1\dots \ell_{k-1}\ell}(w)\ge \Phi_{\ell\ell_1\dots \ell_{k-1}}(w)\ge \Phi_{\ell_1\dots \ell_{k}}(w)$.
	\end{description}
\end{lemma}
\begin{proof}
	For the symmetry property and the first inequality of the monotonicity property, we proceed by induction on $k$. The base case $k=1$ is trivial. For the induction step, it suffices to show that $\Phi_{\ell_1\dots \ell_{k-1} \ell}(w)\ge \Phi_{\ell_1\dots \ell_{k-2}\ell\ell_{k-1}}(w)$, with equality if $\ell=\ell_k$. The lemma then follows by invoking the induction hypothesis on $\tw=w({}\cdot{} \ell_{k-1})\in \cW^{k-1}$, observing that $\min_{X} m_{\tw}(X)=m_{\tw}(\ell_1 \dots \ell_{k-2}\ell_k)$, and that $\Phi_{x_1\dots x_{k-1}\ell_{k-1}}(w)-\Phi_{x_1\dots x_{k-1}}(\tw)=w(\bl_{k-1}^k)$ is a constant function of $x_1,\dots,x_{k-1}$ (and thus $\Phi_{x_1\dots x_{k-1}\ell_{k-1}}(w)$ and $\Phi_{x_1\dots x_{k-1}}(\tw)$ are affected in the same way when $x_1,\dots,x_{k-1}$ are permuted).
	
	Note that only the two terms involving $X_{k-1}$ and $X_k$ in \pref{eq:potLeaves} are affected when the last two leaves are swapped. For two leaves $y$ and $z$, let
	\begin{align*}
		f(y,z):=\min_{\substack{Y\ni z\colon\\Y- z\subseteq\cL- y}}m_w(Y) + \min_{Z\subseteq\cL- z}m_w(Z)
	\end{align*}
	We only need to show that $f(\ell_{k-1},\ell)\ge f(\ell,\ell_{k-1})$, and that this holds with equality if $\ell=\ell_k$. Assume $\ell_{k-1}\ne \ell$ as otherwise there is nothing to show. Then
	\begin{align*}
		f(&\ell_{k-1},\ell)-f(\ell,\ell_{k-1})\\
		&= \min_{\substack{Y_1\ni \ell\colon\\Y_1\subseteq\cL- \ell_{k-1}}}m_w(Y_1) + \min_{Z_1\subseteq\cL- \ell}m_w(Z_1) - \min_{\substack{Y_2\ni \ell_{k-1}\colon\\Y_2\subseteq\cL- \ell}}m_w(Y_2) - \min_{Z_2\subseteq\cL- \ell_{k-1}}m_w(Z_2)\\
		&\ge \min_{Z_1\subseteq\cL- \ell}m_w(Z_1) - \min_{\substack{Y_2\ni \ell_{k-1}\colon\\Y_2\subseteq\cL- \ell}}m_w(Y_2)\\
		&=0,
	\end{align*}
	where the last equation follows by applying \pref{lem:quasiMin} to the quasi-convex function $m_w$. If $\ell=\ell_k$, then the same argument shows that the inequality can be replaced by equality. %\ccnote{\pref{lem:quasiMin} is only proved for sets rather than multisets, but it should easily extend.}
	
	It remains to show the second inequality of the monotonicity property. Due to the symmetry property, and by a renaming of leaves, it suffices to show that $\Phi_{\ell\ell_2\dots \ell_{k}}(w)\ge \Phi_{\ell_1\ell_2\dots \ell_{k}}(w)$. Assume $\ell\ne\ell_1$, otherwise we are done. In \pref{eq:potLeaves}, the only terms affected when the first leaf is replaced are the ones involving $X_0$ and $X_1$. Let $X_0$ and $X_1$ be these sets in $\Phi_{\ell_1\ell_2\dots \ell_{k}}(w)$ and $X_0'$ and $X_1'$ those in $\Phi_{\ell\ell_2\dots \ell_{k}}(w)$. Then $X_0=\ell_1\dots\ell_k$, $X_0'=\ell\ell_2\dots\ell_{k}$, and since $\min_X m_w(X)=m_w(\ell_1\dots\ell_k)$, we can choose $X_1'=\ell_1\dots\ell_k$. Moreover, $X_0'$ satisfies the requirements of $X_1$ (apart from minimality, possibly), hence $m_w(X_1)\le m_w(X_0')$. Thus,
	\begin{align*}
		\Phi_{\ell\ell_2\dots \ell_{k}}(w)-\Phi_{\ell_1\ell_2\dots \ell_{k}}(w)
		&= m_w(X_0')+m_w(X_1')-m_w(X_0)-m_w(X_1)\ge 0.\qedhere
	\end{align*}
\end{proof}

\subsection*{Step 2: $x_1,\dots,x_k$ are indeed leaves}
\begin{lemma}\label{lem:minwminPhi}
	Let $w\in \cW^k$. Let $\ell_1,\dots,\ell_k$ be leaves such that $\min_X m_w(X)=m_w(\ell_1\dots\ell_k)$. Then $\Phi(w)=\Phi_{\ell_1\dots\ell_k}(w)$.
\end{lemma}
\begin{proof}
	By induction on $k$. The base case $k=0$ is trivial. For the induction step, fix $x$ such that $\Phi(w)=\Phi_{x_1\dots x_{k-1}x}(w)$ for some  $x_1,\dots,x_{k-1}$. Consider the function $\tw=w({}\cdot{}x)\in\cW^{k-1}$. By \pref{lem:quasiSub}, $\min m_{\tw}(X)=m_{\tw}(\ell_1\dots\ell_{k}-\ell')$ for some $\ell'\in\ell_1\dots\ell_k$. By \pref{lem:symmetry}, we can assume without loss of generality that $\ell'=\ell_k$, i.e., $\min m_{\tw}(X)=m_{\tw}(\ell_1\dots\ell_{k-1})$. By the induction hypothesis, $\Phi(\tw)=\Phi_{\ell_1\dots\ell_{k-1}}(\tw)$. Hence,
	\begin{align*}
		\Phi(w)&=\min_{x_1\dots x_{k-1}}\Phi_{x_1\dots x_{k-1}x}(w)\\
		&= \min_{x_1\dots x_{k-1}}\Phi_{x_1\dots x_{k-1}}(\tw) + w(\bx^k)\\
		&= \Phi_{\ell_1\dots \ell_{k-1}}(\tw) + w(\bx^k)\\
		&= \Phi_{\ell_1\dots\ell_{k-1}x}(w).
	\end{align*}
	We will now transform the last expression in several steps with the goal of eventually replacing $x$ by $\ell_k$.
	
	Denote by $\ell$ the leaf below $x$. The goal of the following transformations is to replace $x$ by $\ell$. We have
	\begin{align*}
		w(\bx^k)=w(\bl^a\ell^{k-a})+a\cdot x\ell + (k-a)\cdot x\bl
	\end{align*}
	for some $a\in\{0,1,\dots,k\}$. %\ccnote{This might require more explanation or could be extracted as a lemma.}
	
	The symmetry property of \pref{lem:symmetry} allows us to assume that $\ell_1,\dots,\ell_{s-1}$ are all different from $\ell$ and $\ell_{s}=\ell_{s+1}=\dots=\ell_{k-1}=\ell$ for some $s\in\{1,\dots,k\}$.
	
	As an intermediate step, we will show by (backwards) induction on $j=k,k-1,\dots,\max\{s,a\}$ that
	\begin{align}
		\Phi(w)\ge \sum_{i=0}^{j-1}w(\bl_i^i\ell_{i+1}\dots\ell_{k-1}x)+\sum_{i=j}^{k-1}w(\bl_i^{i+1}\ell_{i+1}\dots\ell_{k-1}) \nonumber\\
		\qquad\qquad+ w(\bl^a\ell^{j-a}\ell_j\dots\ell_{k-1})+a\cdot x\ell + (j-a)\cdot x\bl\label{eq:transform1}
	\end{align}
	The base case $k=j$ follows from the previous equation. Suppose now that \pref{eq:transform1} holds for some $j>\max\{s,a\}$. From $\ell_{j-1}=\ell$ we get
	\begin{align*}
		w(\bl_{j-1}^{j-1}\ell_{j}\dots\ell_{k-1}x)&\ge w(\bl_{j-1}^{j}\ell_{j}\dots\ell_{k-1}) - x\bl
		\intertext{and}
		w(\bl^a\ell^{j-a}\ell_j\dots\ell_{k-1})&= w(\bl^a\ell^{j-1-a}\ell_{j-1}\dots\ell_{k-1}).
	\end{align*}
	The induction step of \pref{eq:transform1} follows by plugging these in to \pref{eq:transform1}.
	
	If $a\ge s$, then \pref{eq:transform1} for $j=a$ yields
	\begin{align*}
		\Phi(w)&\ge \sum_{i=0}^{a-1}w(\bl_i^i\ell_{i+1}\dots\ell_{k-1}x)+\sum_{i=a}^{k-1}w(\bl_i^{i+1}\ell_{i+1}\dots\ell_{k-1})+ w(\bl^a\ell_a\dots\ell_{k-1})+a\cdot x\ell\\
		&\ge \sum_{i=0}^{a-1}w(\bl_i^i\ell_{i+1}\dots\ell_{k-1}\ell)+\sum_{i=a}^{k-1}w(\bl_i^{i+1}\ell_{i+1}\dots\ell_{k-1}) + w(\bl^a\ell_a\dots\ell_{k-1})\\
		&= \Phi_{\ell_1\dots\ell_{k-1}\ell}(w),
	\end{align*}
	where we have used that $\ell_i=\ell$ for $i\ge a\ge s$. The lemma then follows from the monotonicity property of \pref{lem:symmetry}.
	
	Otherwise, $a<s$, and from \pref{eq:transform1} for $j=s$ we get
	\begin{align}
		\Phi(w)&\ge \sum_{i=0}^{s-1}w(\bl_i^i\ell_{i+1}\dots\ell_{k-1}x)+\sum_{i=s}^{k-1}w(\bl_i^{i+1}\ell_{i+1}\dots\ell_{k-1}) \nonumber\\
		&\qquad\qquad+ w(\bl^a\ell^{s-a}\ell_s\dots\ell_{k-1})+a\cdot x\ell + (s-a)\cdot x\bl\nonumber\\
		&\ge \sum_{i=0}^{s-a-1}w(\bl_i^i\ell_{i+1}\dots\ell_{k-1}x)+\sum_{i=s-a}^{s-1}w(\bl_i^i\ell_{i+1}\dots\ell_{k-1}\ell)+\sum_{i=s+1}^{k}w(\bl^{i}\ell^{k-i}) \nonumber\\
		&\qquad\qquad+ w(\bl^a\ell^{s-a}\ell_s\dots\ell_{k-1}) + (s-a)\cdot x\bl.\nonumber
	\end{align}
	
	By \pref{cl:quasiPigeon} below, replacing $a$ by $a+1$ does not increase the latter quantity. Inductively we may therefore replace $a$ by $s$ to obtain
	\begin{align*}
		\Phi(w)&\ge \sum_{i=0}^{s-1}w(\bl_i^i\ell_{i+1}\dots\ell_{k-1}\ell)+\sum_{i=s+1}^{k}w(\bl^{i}\ell^{k-i}) \nonumber+ w(\bl^s\ell^{k-s})\\
		&= \Phi_{\ell_1\dots\ell_{k-1}\ell}(w).
	\end{align*}
	The monotonicity property of \pref{lem:symmetry} completes the proof.
\end{proof}

\begin{claim}\label{cl:quasiPigeon}
	Let $0\le a<s\le k$ and $w\in \cW^k$. Let $\ell_{s-a-1},\dots,\ell_{k-1}$ and $\ell$ be leaves such that $\ell_i\ne \ell$ for $i< s$, and let $x$ be a point on the ray of $\ell$. Then 
	\begin{align*}
		&w(\bl_{s-a-1}^{s-a-1}\ell_{s-a}\dots \ell_{k-1} x)+w(\bl^a\ell^{s-a}\ell_s\dots\ell_{k-1})+x\bl\\
		&\qquad \qquad \ge w(\bl_{s-a-1}^{s-a-1}\ell_{s-a}\dots \ell_{k-1} \ell)+w(\bl^{a+1}\ell^{s-a-1}\ell_s\dots\ell_{k-1}).
	\end{align*}
\end{claim}
\begin{proof}
	Consider the bijection from the definition of quasiconvexity between the two configurations on the left hand side.\footnote{We remark that earlier proofs about competitiveness of the work function algorithm only used a weaker form of quasi-convexity and did not actually use the existence of such a bijection.} By the pigeonhole principle, at least one of the $s-a$ copies of $\ell$ in the second configuration maps to some point $p\in\ell_{s-a}\dots\ell_{s-1}x$ in the first configuration. Quasiconvexity gives
	\begin{align*}
		&w(\bl_{s-a-1}^{s-a-1}\ell_{s-a}\dots \ell_{k-1} x)+w(\bl^a\ell^{s-a}\ell_s\dots\ell_{k-1})\\
		&\qquad\qquad \ge w(\bl_{s-a-1}^{s-a-1}\ell_{s-a}\dots \ell_{k-1}\ell x-p)+w(\bl^a\ell^{s-a-1}\ell_s\dots\ell_{k-1}p).
	\end{align*}
	By $1$-Lipschitzness of $w$, we get
	\begin{align*}
		w(\bl_{s-a-1}^{s-a-1}\ell_{s-a}\dots \ell_{k-1}\ell x-p)&\ge w(\bl_{s-a-1}^{s-a-1}\ell_{s-a}\dots \ell_{k-1}\ell) - px
		\intertext{and}
		w(\bl^a\ell^{s-a-1}\ell_s\dots\ell_{k-1}p)&\ge w(\bl^{a+1}\ell^{s-a-1}\ell_s\dots\ell_{k-1})-p\bl.
	\end{align*}
	Since $\ell_i\ne \ell$ for $i<s$ and $p\in \ell_{s-a}\dots\ell_{s-1} x$, the point $x$ lies on the path from $p$ to $\ell$, i.e., $px+x\ell=p\ell$. Equivalently, $x\bl=px+p\bl$. The claim follows by combining these inequalities.
\end{proof}

\subsection*{Step 3: Alternatively, $x_k=r$}

\begin{lemma}\label{lem:multirayPhir}
	For any $w\in\cW^k(r)$, there exist leaves $\ell_1,\dots,\ell_{k-1}$ such that $\Phi(w)=\Phi_{\ell_1\dots\ell_{k-1}r}(w)$.
\end{lemma}
\begin{proof}
	Since $w\in\cW^k(r)$, we can choose $X\in\arg\min_X m_w(X)$ of the form $X=r\ell_2\dots\ell_k$ for $\ell_2,\dots,\ell_k\in\cL$. If $\ell:=\ell_1$ is the leaf of the ray containing $r$, then clearly $\ell_1\dots\ell_k$ is also a minimizer of $m_w$ and $\ell_1\dots\ell_k$ resolves from $\ell_1$. Let $\ell_2,\dots,\ell_k$ be ordered such that $\ell=\ell_1=\dots=\ell_s$ for some $s\ge 1$ and $\ell_i\ne\ell$ for $i>s$.
	
	The main part of this proof is to show that there exists $a\in\{1,\dots,s\}$ such that
	\begin{align}
		w(\bl_{i}^i\ell_{i+1}\dots\ell_k)= \begin{cases}
			w(\bl_{i}^i\ell_{i+2}\dots\ell_kr)+r\ell \qquad&\text{if }i< a\\
			w(\bl_{i}^{i-1}\ell_{i+1}\dots\ell_kr)+r\bl_i \qquad&\text{if }i\ge a.
		\end{cases}\label{eq:mainTask}
	\end{align}
	Before we prove this, let us see why it implies the lemma. By \pref{lem:minwminPhi} and the fact that $\ell=\ell_1=\dots=\ell_a$, we have
	\begin{align*}
		\Phi(w)&=\Phi_{\ell_1\dots\ell_k}(w)\\
		&= \sum_{i=0}^k w(\bl_{i}^i\ell_{i+1}\dots\ell_k)\\
		&= \sum_{i=0}^{a-1}w(\bl_{i}^i\ell_{i+2}\dots\ell_kr) + \sum_{i=a}^kw(\bl_{i}^{i-1}\ell_{i+1}\dots\ell_kr) + a\cdot r\ell + \sum_{i=a}^{k}r\bl_i\\
		&= \sum_{i=0}^{k-1}w(\bl_{i+1}^i\ell_{i+2}\dots\ell_kr) + w(\bl^{a-1}\ell_{a+1}\dots\ell_kr) +  (a-1)\cdot r\ell + \sum_{i=a+1}^{k}\bar r\ell_i + \Delta\\
		&\ge \sum_{i=0}^{k-1}w(\bl_{i+1}^i\ell_{i+2}\dots\ell_kr) + w(\bar r^k)\\
		&= \Phi_{\ell_2\dots\ell_kr}(w).
	\end{align*}
	
	It remains to show \pref{eq:mainTask}. We choose $a$ maximal such that $\bl_{a-1}^{a-1}\ell_{a}\dots\ell_k$ resolves from $\ell$. Recall from the start of this proof that $\ell_1\dots\ell_k$ resolves from $\ell_1$, so $a\ge 1$. Moreover, $a\le s$ since $\ell_i\ne\ell$ for $i>s$. So $a\in\{1,\dots,s\}$ as required. The case ``$i<a$'' of \pref{eq:mainTask} now follows by backwards induction on $i$, where the induction step is due to \pref{lem:resolveMonotone}.
	
	Consider now some $i > s$. Letting $\tw=w(\,\cdot\,\ell_{i+1}\dots\ell_k)\in\cW_{\cL+r}^i$, we have
	\begin{align*}
		w(\bl_i^i\ell_{i+1}\dots\ell_k)&= \tw(\bl_i^i)\\
		&= \min_{X\subseteq\cL+r-\ell_i}m_{\tw}(X) + i(\Delta-c\ell_i),
	\end{align*}
	where the last equation is proved similarly to \pref{lem:potTermLeaves}, but we may allow $X$ to contain the non-leaf $r$ since it is on a different ray than $\ell_i$ (thanks to $i>s$). Since $\min_X m_w(X)=m_w(r\ell_2\dots\ell_k)$, we have $\min_X m_{\tw}(X)=m_{\tw}(r\ell_2\dots \ell_i)$, so by \pref{lem:quasiMin} the minimum under the restriction $X\subseteq\cL+r-\ell_i$ is achieved for some $X$ with $r\in X$.  Thus,
	\begin{align*}
		w(\bl_i^i\ell_{i+1}\dots\ell_k)&= \min_{Y\subseteq\cL-\ell_i}\tw(Yr) - d(Y,c^{i-1}) -r\ell_i+ i(\Delta-c\ell_i)\\
		&= \min_{Y\subseteq\cL-\ell_i}\tw(Yr) + d(Y,\bl_i^{i-1}) +\bl_ir\\
		&\ge \tw(\bl_i^{i-1}r) +\bl_ir\\
		&= w(\bl_i^{i-1}r\ell_{i+1}\dots\ell_k) +r\bl_i.
	\end{align*}
	Note that the inequality between the first and last expression cannot be strict due to $1$-Lipschitzness of $w$, and their equality reveals that $\bl_i^i\ell_{i+1}\dots\ell_k$ resolves from $\bl_i$,	as desired.
	
	Finally, consider $i\in\{a,a+1,\dots,s\}$. Then $\ell_i=\ell$, and we need to show that $\bl^i\ell_{i+1}\dots\ell_k$ resolves from $\bl$. Suppose that it instead resolves from $\ell_h$ for some $h>i$. Since $a\le s$ was chosen maximal, we know that $\ell_h\ne\ell$. By \pref{lem:potTermLeaves},
	\begin{align*}
		w(\bl^i\ell_{i+1}\dots\ell_k)&=w(\bl^i\ell_{i+1}\dots\ell_{h-1}\ell_{h+1}\dots\ell_kr)+r\ell_h\\
		&= \min_{\substack{X\supseteq \ell_{i+1}\dots\ell_{h-1}\ell_{h+1}\dots\ell_kr\\ X-\ell_{i+1}\dots\ell_{h-1}\ell_{h+1}\dots\ell_kr\subseteq\cL-\ell}}m_w(X) + i(\Delta-c\ell) + \sum_{\substack{j=i+1\\j\ne h}}^k c\ell_j + cr+r\ell_h.\\
		&= \min_{Y\not\ni\ell}m_{\tw}(Y) + i(\Delta-c\ell) +r\ell_h,
	\end{align*}
	where $\tw=w(\,\cdot\,\ell_{i+1}\dots\ell_{h-1}\ell_{h+1}\dots\ell_kr)\in\cW_{\cL}^i$. Since $\min_X m_w(X)=m_w(r\ell_2\dots\ell_k)$, we have that $\min_X m_{\tw}(X)=m_{\tw}(\ell_2\dots \ell_i\ell_h)$, so by \pref{lem:quasiMin} the minimum under the restriction $Y\not\ni\ell$ is achieved for some $Y$ with $\ell_h\in Y$. Letting $Y'=Y-\ell_h$, we get
	\begin{align*}
		w(\bl^i\ell_{i+1}\dots\ell_k)&= m_{\tw}(Y) + i(\Delta-c\ell) +r\ell_h\\
		&= w(Y'\ell_{i+1}\dots\ell_kr) -d(Y,c^{i})+ i(\Delta-c\ell) +(rc+c\ell_h)\\
		&= w(Y'\ell_{i+1}\dots\ell_kr) -\sum_{y\in Y'}y\ell+i\Delta-c\ell+rc\\
		&\ge w(Y'\ell_{i+1}\dots\ell_kr) +\sum_{y\in Y'}y\bl+\Delta - r\ell\\
		&\ge w(\bl^{i-1}\ell_{i+1}\dots\ell_kr)+r\bl,
	\end{align*}
	where the second equation uses that $\ell\ne\ell_h$ and therefore $c$ lies on the path from $r$ to $\ell_h$, and the third equation uses that $y$ and $\ell$ are different leaves and therefore $yc+c\ell=y\ell$. Again, the inequality between the first and last expression cannot be strict due to $1$-Lipschitzness of $w$, and their equality reveals that $\bl^i\ell_{i+1}\dots\ell_k$ resolves from $\bl$, completing the proof.
\end{proof}

\begin{theorem}
	WFA is $k$-competitive on multiray spaces.
\end{theorem}
\begin{proof}
	Follows from \pref{lem:multirayPhir} and Corollary~\ref{cor:serverPot}.
\end{proof}

\section{Trees}\label{sec:trees}
Let $V$ be the set of vertices of a tree. In Appendix~\ref{sec:quasiconcavity-trees} we show that a metric $(M,d)$ is a tree if and only if the map $d$ is quasiconcave (i.e., $-d$ is quasiconvex when viewed as a function defined on $2$-point sets). Our proof that WFA is $3$-competitive for the $3$-server problem on trees crucially relies on this property.

We again augment the tree by adding antipodes, as before. We should be careful, though, to apply quasiconcavity of the metric only to distances involving original tree points rather than antipodes.

\begin{lemma}\label{lem:k1arbitraryMinimizer}
	Let $w\in\cW^1$ and $c\in V$. Then $x\in\arg\min w(x)-cx\implies \Phi(w)=\Phi_x(w)$.
\end{lemma}
\begin{proof}
	Let  $x\in\arg\min w(x)-cx$ and let $y$ and $z$ be such that $\Phi(w)=w(y)+w(\by)=w(y)+w(z)+\Delta-yz$. By quasiconcavity, $cx+yz\le cy+xz$ or likewise with $y$ and $z$ reversed; we can assume without loss of generality that the written inequality is the correct one as $y$ and $z$ are symmetric in $\Phi(w)$. Then
	\begin{align*}
		\Phi(w)&=w(y)+w(z)+\Delta-yz\\
		&\ge w(y)+w(z)+\Delta+cx-cy-xz\\
		&\ge w(x)+w(z)+\Delta-xz\\
		&\ge w(x)+w(\bx)\\
		&=\Phi_x(w).
	\end{align*}
	where we have used that $w(y)-cy\ge w(x)-cx$ by choice of $x$.
\end{proof}

\begin{lemma}\label{lem:treeSwapx12}
	Let $w\in\cW^k$. There exist $x_1,\dots,x_k$ such that $\Phi(w)=\Phi_{x_1\dots x_k}(w)=\Phi_{x_2x_1x_3\dots x_k}(w)$ and a copy of $\bx_2$ in $\bx_2^2x_3\dots x_k$ resolves to $x_1$.
\end{lemma}
\begin{proof}
	It suffices to show the lemma for the case $k=2$ (otherwise, consider $w(\,\cdot\,x_3\dots x_k)\in\cW^2$). Fix $x_2$ such that $\Phi(w)=\Phi_{x_1x_2}(w)$ for some $x_1$. By \pref{lem:quasiGreedy} applied to $X\mapsto w(X)-d(X,x_2^2)$ and $A=x_2$, we conclude that a copy of $\bx_2$ in $\bx_2^2$ resolves to some $x_1\in\arg\min_x w(xx_2)-xx_2$. By \pref{lem:k1arbitraryMinimizer} applied to $w(\,\cdot\,x_2)\in\cW^1$ and $c=x_2$, it holds that $\Phi(w)=\Phi_{x_1x_2}(w)$. Finally, note that we can swap the order of $x_1$ and $x_2$ because
	\begin{align*}
		w(\bx_1x_2)+w(\bx_2^2)=w(\bx_1x_2)+w(\bx_2x_1)+x_1\bx_2\ge w(\bx_1^2)+w(\bx_2x_1).&\qedhere
	\end{align*}
\end{proof}

\begin{lemma}\label{lem:treeResolveLastTwo}
	Let $w\in \cW^k(r)$ and $x, y\in V$. Then
	\begin{align*}
		w(\bx^{k-1}y)+w(\by^k)+(k-1)ry \ge w(\bx^{k-1}r)+w(\br^k).
	\end{align*}
\end{lemma}
\begin{proof}
	If $\bx^{k-1}y$ resolves from $y$, the statement follows directly from $1$-Lipschitzness of $w$. So assume
	\begin{align*}
		w(\bx^{k-1}y)=w(\bx^{k-2}yr)+\Delta-rx.
	\end{align*}
	We have
	\begin{align*}
		w(\by^k)=w(a_1\dots a_{k-1}r)+k\Delta-\sum_{i=1}^{k-1}a_iy-ry
	\end{align*}
	for some $a_1,\dots,a_{k-1}\in V$. If $rx+a_iy\le xy+a_ir$ for some $i$, then
	\begin{align*}
		w(\bx^{k-1}y)+w(\by^k)+(k-1)ry &= w(\bx^{k-2}yr)+w(\by^{k-1}a_i)+2\Delta-rx-a_iy+(k-1)ry\\
		&\ge w(\bx^{k-2}yr)+w(\br^{k-1}a_i)+\bx y+a_i\br\\
		&\ge w(\bx^{k-1}r)+w(\br^k).
	\end{align*}
	Otherwise, by quasiconcavity of the distance we have $rx+a_iy=ry+a_ix$ for \emph{every} $i$. Then
	\begin{align*}
		w(\bx^{k-1}&y)+w(\by^k)+(k-1)ry\\
		&= w(\bx^{k-2}yr)+w(a_1\dots a_{k-1}r)+(k+1)\Delta-rx-\sum_{i=1}^{k-1}a_iy+(k-2)ry\\
		&= w(\bx^{k-2}yr)+w(a_1\dots a_{k-1}r)+(k+1)\Delta-ry-\sum_{i=1}^{k-1}a_ix+(k-2)rx\\
		&\ge w(\br^k)+w(\bx^{k-1}r).\qedhere
	\end{align*}
\end{proof}

\begin{theorem}
	WFA is $3$-competitive for $3$ servers on trees.
\end{theorem}
\begin{proof}
	Let $w\in\cW^3(r)$ and let $x_1,x_2,x_3$ be such that $\Phi(w)=\Phi_{x_1x_2x_3}(w)$. If $x_1x_2x_3$ resolves from $x_1$ or $x_2$, then \pref{lem:treeSwapx12} allows us to assume that it resolves from $x_1$. But then $\Phi(w)=\Phi_{x_1x_2x_3}(w)\ge\Phi_{rx_2x_3}$ by $1$-Lipschitzness of $w$. Then, \pref{lem:push3} shows that $\Phi(w)=\Phi_{yzr}(w)$ for some $y$ and $z$, implying $3$-competitiveness. We can therefore assume that \emph{$x_1x_2x_3$ resolves from $x_3$}.
	
	If $\bx_1x_2x_3$ resolves from $\bx_1$, then since $\Phi_{x_1x_2x_3}=\Phi_{\bx_1x_2x_3}$ the same argument implies $3$-competitiveness. So $\bx_1x_2x_3$ resolves from $x_2$ or $x_3$. If it resolves from $x_3$, then
	\begin{align*}
		\Phi(w)&=w(x_1x_2r)+w(\bx_1x_2r)+w(\bx_2^2x_3)+w(\bx_3^3)+2rx_3\\
		&\ge w(x_1x_2r)+w(\bx_1x_2r)+w(\bx_2^2r)+w(\br^3)\\
		&=\Phi_{x_1x_2r}(w),
	\end{align*}
	where the inequality is due to \pref{lem:treeResolveLastTwo}. So assume \emph{$\bx_1x_2x_3$ resolves from $x_2$}.
	If $\bx_2^2x_3$ resolves from $x_3$, then \pref{lem:treeSwapx12} allows us to assume that a copy of $\bx_2$ in $\bx_2^2x_3$ resolves to $x_1$, so
	\begin{align*}
		w(\bx_1x_2x_3)+w(\bx_2^2x_3)&=w(\bx_1x_3r)+w(\bx_2x_1r)+rx_2+\bx_2x_1+rx_3\\
		&\ge w(\bx_1x_3r)+w(\bx_2x_1r)+r\bx_1+rx_3\\
		&\ge w(\bx_1^2x_3)+w(\bx_2x_1x_3)
	\end{align*}
	Here, note that all inequalities must be equality because otherwise we would have shown $\Phi_{x_2x_1x_3}(r)<\Phi_{x_1x_2x_3}$. So $\bx_2x_1x_3$ resolves from $x_3$, and $3$-competitiveness follows symmetrically to the case that $\bx_1x_2x_3$ resolves from $x_3$ with the roles of $x_1$ and $x_2$ reversed. So we can assume that \emph{$\bx_2^2x_3$ resolves from $\bx_2$}.
	
	For these resolutions, we can conclude $3$-competitiveness using only quasiconvexity of $w$. We have
	\begin{align*}
		\Phi(w)&=w(x_1x_2r)+w(\bx_1 x_3r)+w(\bx_2x_3r)+w(\bx_3^2r)+rx_3+rx_2+r\bx_2+r\bx_3\\
		&=w(x_1x_2r)+w(\bx_1 x_3r)+w(\bx_2x_3r)+w(\bx_3^2r)+2\Delta
	\end{align*}
	By quasiconvexity, $w(x_1x_2r)+w(\bx_1 x_3r)$ is lower bounded by $w(x_1\bx_1r)+w(x_2 x_3r)$ or $w(x_1x_3r)+w(\bx_1 x_2r)$. In the first case, using $\Delta=\br x_1+\br \bx_1$, we get
	\begin{align*}
		\Phi(w)&\ge w(x_1\bx_1r)+w(x_2 x_3r)+w(\bx_2x_3r)+w(\bx_3^2r)+\Delta+\br x_1+\br\bx_1\\
		&\ge w(\br^3)+w(x_2 x_3r)+w(\bx_2x_3r)+w(\bx_3^2r)\\
		&= \Phi_{x_2x_3r}(w).
	\end{align*}
	In the second case,
	\begin{align*}
		\Phi(w)&\ge w(x_1x_3r)+w(\bx_1 x_2r)+w(\bx_2x_3r)+w(\bx_3^2r)+2\Delta\\
		&\ge w(x_1x_3r)+w(\bx_1 x_2x_3)+w(\bx_2x_3r)+w(\bx_3^2r)+2\Delta-rx_3\\
		&\ge w(x_1x_3r)+w(\bx_1 x_3r)+w(\bx_2x_3r)+w(\bx_3^2r)+2\Delta-rx_3+rx_2\\
		&\ge w(x_1x_3r)+w(\bx_1 x_3r)+w(\br^3)+w(\bx_3^2r)\\
		&= \Phi_{x_1x_3r}(w),
	\end{align*}
	where the third inequality reuses the fact that $\bx_1x_2x_3$ resolves from $x_2$. Again, we conclude $3$-competitiveness.
\end{proof}
\section{Non-laziness of the worst-case adversary on the circle}\label{sec:nonlazy}
\begin{theorem}\label{thm:Nonlazy}
	For $k=3$ servers on the circle, there exists a reachable work function from where the worst-case adversarial continuation of the request sequence is not lazy. More precisely, there exists a request sequence such that the induced work functions $w_{t}$ and $w_{t+1}$ after time steps $t$ and $t+1$ and the WFA configuration $C_{t}$ after time step $t$ satisfy $w_{t+1}(C_{t})-w_t(C_t) > \Phi(w_{t+1}) - \Phi(w_t)$.
\end{theorem}

In other words, the extended cost is strictly greater than the change in potential. Due to the interpretation of our potential (Section~\ref{sec:interpretation}), this means that the worst-case continuation of the request sequence after time $t$ is not lazy.

If Theorem~\ref{thm:Nonlazy} could be strengthened such that the request sequence to reach $w_t$ has extended cost equal to its induced change in potential, then this would disprove the premise of the extended cost lemma (because one could create a cyclic request sequence where extended cost is always at least the change in potential and exceeds it infinitely often; we remark that one can go from a cone work function to any other cone via a request sequence whose extended cost equals its potential change). Note that Theorem~\ref{thm:Nonlazy} holds even if in the extended cost $\max_X w_{t+1}(X)-w_t(X)$ we replace $X$ by the configuration $C_t$ of WFA at time $t$. The significance of this is that the sum of the terms $w_{t+1}(C_{t})-w_t(C_t)$ over all time steps is \emph{equal} to the sum of WFA's cost and the optimal offline cost (up to a bounded additive error). Thus, proving violation of the premise of the extended cost lemma with $X$ replaced by $C_t$ would imply that WFA's competitive ratio is strictly greater than $k$.

The proof of Theorem~\ref{thm:Nonlazy} is based on a tight connection between the $k$-server problem and the ``easy'' version of the $k$-taxi problem that we had observed in \cite{CoesterK19}. The $k$-taxi problem is the generalization of the $k$-server problem where each request is not a single point, but a pair $(s,t)$ of two points, representing the start $s$ and destination $t$ of a taxi request. To serve it, the algorithm has to select a server that first goes to $s$ and then to $t$. In the ``easy'' version relevant for us, the cost is defined as the total distance traveled by servers.\footnote{In contrast, the ``hard'' $k$-taxi problem defines the cost as only the overhead distance traveled while \emph{not} carrying a passenger, i.e., the distance from $s$ to $t$ is excluded from the cost (motivated by the fact that any algorithm has to travel this distance).} As we showed in \cite{CoesterK19}, the easy $k$-taxi problem has exactly the same competitive ratio as the $k$-server problem. The idea of this reduction is that a $k$-taxi request $(s,t)$ can be simulated by a sequence of many $k$-server requests along the shortest path from $s$ to $t$. We extend this idea here to show that we can use $k$-taxi requests to reach work functions that are arbitrarily close to work functions that are also reachable via $k$-server requests.

\subsection{\texorpdfstring{Approximate $k$-server work functions via $k$-taxi requests}{Approximate k-server work functions via k-taxi requests}}

For a work function $w$, we denote by $w\land(s,t)$ the updated work function when simulating a $k$-taxi request $(s,t)$ via $k$-server requests. More precisely, for any configuration $C$ we define $w\land(s,t)(C)$ to be the limit, as $m\to\infty$, of $w\land r_1\land r_2\land \dots \land r_m (C)$, where $r_1,r_2,\dots,r_m$ are equally spaced points along the shortest path from $s=r_1$ to $t=r_m$. Note that for $s$ and $t$ on the circle, the shortest path and hence $r_1,\dots,r_m$ are unique for fixed $m$ unless $s$ and $t$ are antipodes of each other, in which case we may choose them along any of the two shortest paths, say clockwise.

By the next lemma, simulating the taxi request $(s,t)$ via $k$-server requests has the same effect as issuing a $k$-server request at $s$, and then replacing $s$ by $t$ in each support configuration and increasing the supporting work function values by $d(s,t)$.

\begin{lemma}
	Let $M$ be the circle, $s,t\in M$ and $w\in \cW_{M}$ be a work function. Then
	\begin{align}
		\supp(w\land(s,t))=\{S-s+t\colon S\in \supp(w\land s)\}\label{eq:taxiSupp}
	\end{align}
	and for each $S\in \supp(w\land s)$,
	\begin{align}
		w\land(s,t)(S-s+t)=w(S)+d(s,t).\label{eq:taxiVal}
	\end{align}
\end{lemma}
\begin{proof}
	We will show that for each configuration $C$,
	\begin{align}
		w\land (s,t)(C)=\min_{S\in\supp(w\land s)}w(S)+d(s,t)+d(S-s+t,C).\label{eq:taxiToShow}
	\end{align}
	Let us first argue why this implies the lemma. Taking $C=S'-s+t$ for $S'\in\supp(w\land s)$, equation \eqref{eq:taxiToShow} implies equation \eqref{eq:taxiVal} because
	\begin{align*}
		w\land (s,t)(S'-s+t) &=\min_{S\in\supp(w\land s)}w(S)+d(s,t)+d(S-s+t,S'-s+t)\\
		&=\min_{S\in\supp(w\land s)}w(S)+d(s,t)+d(S,S')\\
		&= w(S')+d(s,t)
	\end{align*}
	where the last equation uses that by $1$-Lipschitzness of $w$, the minimum is achieved for $S=S'$.
	
	Now, substituting \eqref{eq:taxiVal} in \eqref{eq:taxiToShow}, we get
	\begin{align*}
		w\land (s,t)(C)=\min_{S\in\supp(w\land s)}w\land(s,t)(S-s+t)+d(S-s+t,C),
	\end{align*}
	which yields the inclusion ``$\subseteq$'' in \eqref{eq:taxiSupp}. If this inclusion were strict, then there would exist $S, S'\in\supp(w\land s)$ such that $S'-s+t$ is supported by $S-s+t$ in $w\land(s,t)$. But then
	\begin{align*}
		d(S,S')&= d(S-s+t,S'-s+t)\\
		&= w\land(s,t)(S'-s+t) - w\land(s,t)(S-s+t)\\
		&= w(S')-w(S)\\
		&= w\land s(S')-w\land s(S)
	\end{align*}
	where the penultimate equation uses \eqref{eq:taxiVal} and the last equation uses $s\in S\cap S'$. But this would mean that $S'\notin \supp(w\land s)$, a contradiction.
	
	It remains to show \eqref{eq:taxiToShow}.
	
	Let $r_1,\dots, r_m$ be equally spaced from $s$ to $t$ and let $\epsilon=d(s,t)/(m-1)$ be the distance between any two adjacent $r_i$ and $r_{i+1}$. Note that $w\land r_1\land r_2\land \dots \land r_m(C)$ is the minimum, over all $S\in \supp(w\land r_1)=\supp(w\land s)$, of the sum of $w(S)=w\land s(S)$ plus the cheapest way of serving the requests $r_1, \dots r_m$ starting from $S$ and ending at $C$. One way of serving the requests $r_1,\dots,r_m$ starting from $S$ and ending at $C$ is to first take the server at $s$ and move it along the shortest path from $s$ to $t$ for cost $d(s,t)$ (which serves all requests $r_1,\dots, r_m$ and reaches configuration $S-s+t$) and then move to configuration $C$ for an additional cost $d(S-s+t,C)$. This shows the direction ``$\le$'' of \eqref{eq:taxiToShow}.
	
	However, the cheapest way of serving $r_1,\dots,r_m$ starting from $S$ and ending at $C$ might use several different servers to serve the $r_1,\dots,r_m$. To obtain the direction ``$\ge$'', we will show now that this can be at most $2k\epsilon$ cheaper, which is negligible as $m\to\infty$.
	
	If two different servers $a$ and $b$ are used to serve $r_i$ and $r_{i+1}$, then server $a$ will not be used to serve $r_j$ for \emph{any} $j>i$. This is because $a$ would have to move past the location of server $b$ to do so, but then it is at least as good to use server $b$ instead. Thus, $r_1,\dots,r_m$ can be partitioned into $j$ contiguous subsequences for some $j\le k$ such that the same server is used to serve the requests within each contiguous subsequence. Let $a_1,\dots,a_j$ be the locations of these servers in $S$ (in the order in which they are used, so $a_1=s$) and let $s_i$ and $t_i$ be the first and last request, respectively, of the contiguous subsequence of $r_1,\dots,r_m$ that is served by $a_i$. Also let $t_0=s$, so that $d(s_i,t_{i-1})\le\epsilon$ for each $i=1,\dots,j$. The $i$th server pays movement cost $d(a_i,s_i)+d(s_i,t_i)\ge d(a_i,t_{i-1}) + d(t_{i-1},t_i)-2\epsilon$ to serve its requests, and finally cost $d(S-a_1\dots a_j + t_1\dots t_j, C)$ is paid to reach configuration $C$. So the overall cost of serving $r_1,\dots,r_m$ starting from configuration $S$ and ending in $C$ is at least
	\begin{align*}
		&\sum_{i=1}^j\left[d(a_i,t_{i-1}) + d(t_{i-1},t_i)-2\epsilon\right] + d(S-a_1\dots a_j + t_1\dots t_j, C)\\
		&\ge d(t_0,t_j)-2k\epsilon+d(a_1,t_{0}) + d(S-a_1+t_j, C)\\
		&= d(s,t)-2k\epsilon+ d(S-s+t, C)
	\end{align*}
	where the inequality follows from several applications of the triangle inequality and the equation holds because $t_0=a_1=s$ and $t_j=t$. Since $\epsilon\to0$ as $m\to\infty$, we conclude direction ``$\ge$'' of equation~\eqref{eq:taxiToShow}.
\end{proof}

\subsection{The counterexample}
Consider the interval $[0,8)$ equipped with the circle metric $d$, i.e., $d(a,b)=\min\{b-a,8+a-b\bmod 8\}$ for $a\le b\in[0,8)$. We describe a sequence of mixed $k$-taxi and $k$-server requests to reach work functions $w_t$ and $w_{t+1}$ satisfying \pref{thm:Nonlazy} with the stronger inequality $w_{t+1}(C_{t})-w_t(C_t) \ge \Phi(w_{t+1}) - \Phi(w_t)+1$. Since work functions reachable by $k$-taxi requests are approximated arbitrarily well by work functions reachable by $k$-server requests, this shows that \pref{thm:Nonlazy} holds also for work functions reachable via $k$-server requests only.

\begin{figure}
	\centering
	\includegraphics[scale=1.15]{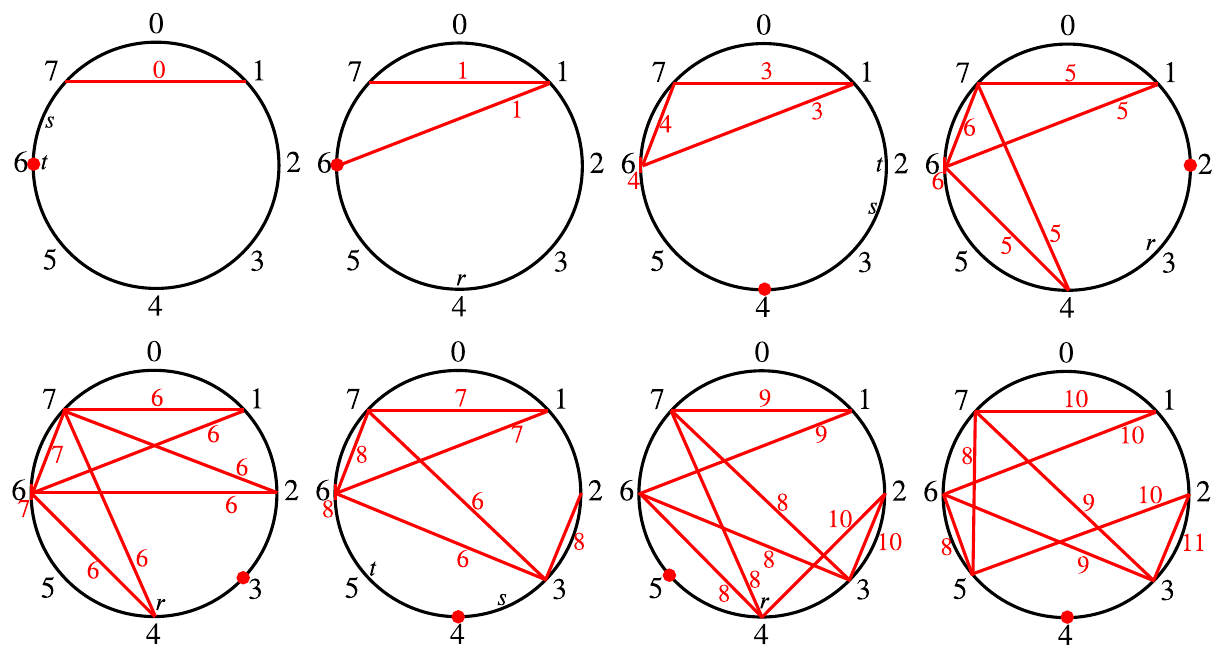}
	\caption{Evolution of the work function for request sequence yielding the counterexample to laziness of the adversary. Each of the eight images depicts one work function, starting with the work function before the first request (top left) and ending with work functions $w_{t}$ and $w_{t+1}$ satisfying $w_{t+1}(C_{t})-w_t(C_t) \ge \Phi(w_{t+1}) - \Phi(w_t)+1$ (last two images in the bottom row). In each image, the red dots indicates the location of the last request, which is contained every support configuration of the current work function. (In the initial configuration $\{1,6,7\}$, we view $6$ as the last request.) For every support configuration, there is a red line connecting the two other points of the configuration, and the red number next to the red line is the work function value of that configuration. In the third to sixth images, the very short line at $6$ indicates the support configuration containing two copies of $6$ and the last request. Only support configurations are depicted. In each but the last circle, either two points are marked as $s$ and $t$ or one point is marked as $r$, indicating that the next work function is reached by the $k$-taxi request $(s,t)$ or the $k$-server request $r$.}
	\label{fig:circle}
\end{figure}

Starting from an initial configuration $\{1,6,7\}$, the work function $w_t$ is reached by issuing the following requests:
\begin{enumerate}[(a)]
	\item $(6.5,6)$
	\item $4$
	\item $(2.5,2)$
	\item $3$
	\item $4$
	\item $(3.5,5)$
\end{enumerate}
where single numbers denote $k$-server requests and pairs denote $k$-taxi requests. Work function $w_{t+1}$ is reached by issuing one more request at $4$. The corresponding evolution of the work function is depicted in Figure~\ref{fig:circle}. For the corresponding $k$-server request sequence where taxi requests are simulated by sufficiently many equally spaced requests as described above, one can verify that WFA serves every request using the same server that initially resides at $6$. (This assumes that ties are broken in favor of this server. By a tiny perturbation of the request sequence, one can force any instantiation of WFA to break in this way.) Thus, when work function $w_t$ is reached, WFA is in configuration $C_t=\{1,5,7\}$, and from the last two images in Figure~\ref{fig:circle} we see that
\begin{align*}
	w_{t+1}(C_t)-w_t(C_t)=11-9 = 2.
\end{align*}
Moreover, we have
\begin{align*}
	&\Phi(w_{t+1})\\
	&\le\Phi_{572}(w_{t+1})\\
	&=w_{t+1}(572)+w_{t+1}(172)+w_{t+1}(332)+w_{t+1}(666)\\
	&\le \left[w_{t+1}(574)+d(2,4)\right]+\left[w_{t+1}(174)+d(2,4)\right]+\left[w_{t+1}(432)+d(3,4)\right]+\left[w_{t+1}(654)+d(5,6)+d(4,6)\right]\\
	&= \left[8+2\right]+\left[10+2\right]+\left[11+1\right]+\left[8+1+2\right]\\
	&=45.
\end{align*}
In $\Phi(w_t)=\min_{x_1,x_2,x_3}\Phi_{x_1x_2x_3}(w_t)$, the minimum is achieved for $x_1=4$, $x_2=5$, $x_3=6$ (and also several other choices). We omit a proof of this, which we found using a computer. A similar calculation then shows that, $\Phi(w_{t})=44$. Thus, $w_{t+1}(C_t)-w_t(C_t)\ge\Phi(w_{t+1})-\Phi(w_{t})+1$, as claimed.

We remark that up to symmetry and shift by an additive constant, for $k=3$ there exist over 280,000 different work functions reachable using taxi requests whose destinations are at $8$ equally spaced points on a circle and whose starts are at any of the same 8 points or the $8$ intermediate points (such as $6.5$, $2.5$ and $3.5$ above). Among these over 280,000 work functions, the above pair of $w_t$ and $w_{t+1}$ is the \emph{only} counterexample to laziness of the adversary. Using only $k$-server requests and no $k$-taxi requests, we were unable to find any counterexamples for $n$ equally spaced points on the circle for the values of $n$ that were computationally feasible for us to try. Of course, though, our approximability argument of $k$-taxi requests via $k$-server requests implies that such counterexamples do exist for $n$ sufficiently large. Given the rarity of these counterexamples, it is not surprising that Chrobak and Larmore~\cite{ChrobakL92} who reported testing their conjecture on tens of thousands of small metric spaces in the early 90s did not find any counterexample.
\section{Conclusion}
Our potential gives a unified perspective on all cases where WFA is known to be $k$-competitive. Unlike previous potentials, which were specific to their special case and had no clear intuition, our potential has a natural interpretation as capturing a lazy adversary. We remark that beyond the cases proved in this paper, our potential also proves $k$-competitiveness on $6$-point metric spaces. Since work functions, the WFA, and the generalized WFA are central to various online problems, similar potential functions may also prove useful to analyze different problems.

Since it was a major belief that a lazy adveresary would capture the worst case, our insights yield a qualitative explanation of the shortcomings of previous approaches and may point in a direction to overcome these shortcomings.

We are puzzled by the role of quasiconcavity of tree metrics. Does it have any deeper connection to the quasiconvexity property of work functions making it crucial for the existence of $k$-competitive algorithms? While the $k$-server problem is also known to be $k$-competitive in some cases without quasiconcavity (such as $k=2$ and $n=k-2$), the reason for this might simply be due to the fact that the subspaces relevant in all proof steps are small (note that any $3$-point metric is quasiconcave).

\bibliography{bibliography}{}

\newcommand{\etalchar}[1]{$^{#1}$}
\begin{thebibliography}{BBMN15}

\bibitem[AGGT20]{ArgueGGT20}
C.~J. Argue, Anupam Gupta, Guru Guruganesh, and Ziye Tang.
\newblock Chasing convex bodies with linear competitive ratio.
\newblock In {\em Proceedings of the 2020 {ACM-SIAM} Symposium on Discrete
  Algorithms, {SODA} 2020, Salt Lake City, UT, USA, January 5-8, 2020}, pages
  1519--1524, 2020.

\bibitem[BBMN15]{BansalBMN15}
Nikhil Bansal, Niv Buchbinder, Aleksander Madry, and Joseph Naor.
\newblock A polylogarithmic-competitive algorithm for the \emph{k}-server
  problem.
\newblock {\em J. {ACM}}, 62(5):40:1--40:49, 2015.

\bibitem[BCL02]{BeinCL02}
Wolfgang~W. Bein, Marek Chrobak, and Lawrence~L. Larmore.
\newblock The 3-server problem in the plane.
\newblock {\em Theor. Comput. Sci.}, 289(1):335--354, 2002.

\bibitem[BCL{\etalchar{+}}18]{BCLLM18}
S{\'{e}}bastien Bubeck, Michael~B. Cohen, Yin~Tat Lee, James~R. Lee, and
  Aleksander Madry.
\newblock $k$-server via multiscale entropic regularization.
\newblock In {\em Proceedings of the 50th Annual {ACM} {SIGACT} Symposium on
  Theory of Computing, {STOC} 2018}, pages 3--16, 2018.

\bibitem[BE98]{BorodinE98}
Allan Borodin and Ran El{-}Yaniv.
\newblock {\em Online computation and competitive analysis}.
\newblock Cambridge University Press, 1998.

\bibitem[BEK17]{BansalEK17}
Nikhil Bansal, Marek Eli{\'{a}}s, and Grigorios Koumoutsos.
\newblock Weighted k-server bounds via combinatorial dichotomies.
\newblock In {\em 58th {IEEE} Annual Symposium on Foundations of Computer
  Science, {FOCS} '17}, 2017.

\bibitem[BK04]{BartalK04}
Yair Bartal and Elias Koutsoupias.
\newblock On the competitive ratio of the work function algorithm for the
  k-server problem.
\newblock {\em Theoretical Computer Science}, 324(2-3):337--345, September
  2004.

\bibitem[BLS92]{BorodinLS92}
Allan Borodin, Nathan Linial, and Michael~E. Saks.
\newblock An optimal on-line algorithm for metrical task system.
\newblock {\em J. ACM}, 39(4):745--763, October 1992.

\bibitem[Bur96]{Burley96}
William~R. Burley.
\newblock Traversing layered graphs using the work function algorithm.
\newblock {\em J. Algorithms}, 20(3):479--511, 1996.

\bibitem[CK19]{CoesterK19}
Christian Coester and Elias Koutsoupias.
\newblock The online $k$-taxi problem.
\newblock In {\em Proceedings of the 51st Annual {ACM} {SIGACT} Symposium on
  Theory of Computing, {STOC} '19}, 2019.

\bibitem[CKPV91]{ChrobakKPV91}
Marek Chrobak, Howard Karloff, Tom Payne, and Sundar Vishwanathan.
\newblock New results on server problems.
\newblock {\em SIAM Journal on Discrete Mathematics}, 4(2):172--181, 1991.

\bibitem[CL91]{ChrobakL91}
Marek Chrobak and Lawrence~L. Larmore.
\newblock An optimal on-line algorithm for k servers on trees.
\newblock {\em SIAM Journal on Computing}, 20(1):144--148, 1991.

\bibitem[CL92]{ChrobakL92}
Marek Chrobak and Lawrence~L Larmore.
\newblock The server problem and on-line games.
\newblock In {\em On-line Algorithms, volume 7 of DIMACS Series in Discrete
  Mathematics and Theoretical Computer Science}. Citeseer, 1992.

\bibitem[DW90]{Dress1990}
Andreas~W.M. Dress and Walter Wenzel.
\newblock {Valuated matroids: a new look at the greedy algorithm}.
\newblock {\em Applied Mathematics Letters}, 3(2):33--35, jan 1990.

\bibitem[FRR94]{FiatRR94}
Amos Fiat, Yuval Rabani, and Yiftach Ravid.
\newblock Competitive k-server algorithms.
\newblock {\em J. Comput. Syst. Sci.}, 48(3), 1994.

\bibitem[KJC82]{kelso1982job}
Alexander~S Kelso~Jr and Vincent~P Crawford.
\newblock Job matching, coalition formation, and gross substitutes.
\newblock {\em Econometrica: Journal of the Econometric Society}, pages
  1483--1504, 1982.

\bibitem[Kou99]{Koutsoupias99}
Elias Koutsoupias.
\newblock Weak adversaries for the k-server problem.
\newblock In {\em 40th Annual Symposium on Foundations of Computer Science,
  {FOCS} '99}, pages 444--449, 1999.

\bibitem[Kou09]{Koutsoupias09}
Elias Koutsoupias.
\newblock The k-server problem.
\newblock {\em Computer Science Review}, 3(2):105--118, 2009.

\bibitem[KP95]{KoutsoupiasP95}
Elias Koutsoupias and Christos~H. Papadimitriou.
\newblock On the k-server conjecture.
\newblock {\em J. {ACM}}, 42(5), 1995.

\bibitem[KP96]{KoutsoupiasP96}
Elias Koutsoupias and Christos Papadimitriou.
\newblock The 2-evader problem.
\newblock {\em Information Processing Letters}, 57(5):249--252, 1996.

\bibitem[Lee18]{Lee18}
James~R. Lee.
\newblock Fusible {HSTs} and the randomized k-server conjecture.
\newblock In {\em Proceedings of the 59th Annual IEEE Symposium on Foundations
  of Computer Science}, FOCS '18, pages 438--449, 2018.

\bibitem[MMS88]{ManasseMS88}
Mark~S. Manasse, Lyle~A. McGeoch, and Daniel~Dominic Sleator.
\newblock Competitive algorithms for on-line problems.
\newblock In {\em Proceedings of the 20th Annual {ACM} Symposium on Theory of
  Computing, {STOC '88}}, 1988.

\bibitem[Mur03]{Murota2003c}
Kazuo Murota.
\newblock Society for Industrial and Applied Mathematics, jan 2003.

\bibitem[Sel20]{Sellke20}
Mark Sellke.
\newblock Chasing convex bodies optimally.
\newblock In {\em Proceedings of the 2020 {ACM-SIAM} Symposium on Discrete
  Algorithms, {SODA} '20}, pages 1509--1518, 2020.

\bibitem[Sit14]{Sitters14}
Ren{\'{e}} Sitters.
\newblock The generalized work function algorithm is competitive for the
  generalized 2-server problem.
\newblock {\em {SIAM} J. Comput.}, 43(1), 2014.

\end{thebibliography}
\bibliographystyle{alpha}

\appendix
%\part{Appendix}
\newpage
\section{Appendix}

In the following two subsections, we use our potential to prove $k$-competitiveness of WFA for the cases $k=2$, $k=n-1$ and $k=n-2$. We omit a formal proof that our potential also works for the special case of $3$ servers in the Manhattan plane~\cite{BeinCL02}, as it requires substantial case analysis. Alternatively, careful inspection of the results proved in~\cite{BeinCL02} reveals that our potential is in fact equal to the potential defined there for this special case.

In the last subsection, we prove that a metric is a tree if and only if it is quasiconcave.

\subsection{\texorpdfstring{$k=2$}{k=2}}

We consider here the relatively simple case of $k=2$. There are multiple proofs of
this special case, but we provide one here that uses the potential, for completeness.

\begin{lemma}
	For any $w\in \cW^2(r)$, there exists $x_1$ such that $\Phi(w)=\Phi_{x_1 r}(w)$.
\end{lemma}
\begin{proof}
	Suppose $\Phi(w)=\Phi_{x_1 x_2}(w)$. We will show that $\Phi_{x_1x_2}(w)\geq
	\min\{\Phi_{x_1 r}(w), \Phi_{x_2 r}(w)\}$, which will establish the lemma.
	
	\paragraph{Case 1:} $w(\bx_1 x_2)=w(\bx_1 r)+rx_2$ and $w(x_1 x_2)=w(x_1r)+rx_2$.
	
	Then
	\begin{align*}
	\Phi_{x_1x_2}(w)&=w(\bx_2^2)+w(\bx_1 x_2)+
	w(x_1x_2)\\
	&=w(\bx_2^2)+w(\bx_1 r)+rx_2+w(x_1r)+rx_2\\
	&\geq
	w(\br^2)+w(\bx_1 r)+w(x_1r)\\
	&=\Phi_{x_1 r}(w).
	\end{align*}
	
	\paragraph{Case 2:} Otherwise, $w(\bx_1 x_2)=w(r x_2)+r\bx_1$ or $w(x_1 x_2)=w(r x_2)+r x_1$.
	
	Since $\Phi_{\bx_1x_2}=\Phi_{x_1x_2}$, we may assume that the first of these two equations holds. Then since $w(\bx_2^2)=w(r\bx_2)+r\bx_2$ and $w(\br^2)\leq w(x_1x_2)+\br
	x_1+\br x_2=w(x_1x_2)+r\bx_1+r\bx_2$, we get 
	\begin{align*}
	\Phi_{x_1x_2}(w)&=w(\bx_2^2)+w(\bx_1 x_2)+
	w(x_1x_2)\\
	&=w(r\bx_2)+r\bx_2+w(rx_2)+r\bx_1+w(x_1x_2)\\
	&\geq
	w(r\bx_2)+w(rx_2)+w(\br^2)\\
	&=\Phi_{x_2 r}(w).\qedhere
	\end{align*}
\end{proof}

\subsection{\texorpdfstring{$k=n-1$ and $k=n-2$}{k=n-1 and k=n-2}}

\newcommand{\dw}{\hat w}

In this subsection we consider the case of $k=n-2$ servers (equivalent to the case of
2 evaders). We show below that the potential works for this case. Actually, the
potential works also for the case $k=n-1$ (1-evader), because it follows immediately
from its definition that the value of the potential is equal to the sum of the $n$
values of the work function, one value for each point of the metric space; therefore
any permutation of the points gives the same value and we can assume that $r$ is the
last point.

We now turn our attention to the case of 2 evaders ($k=n-2$). For this case, it is
easy to argue that the evader potential $\hPhi_y(w)$ is the value of a minimum
spanning tree with weights $\dw(x,y)+xy$, where $\dw(x,y)=w(M\setminus
\{x,y\})$. Notice first that the term
$$\min_{\substack{C\subseteq\{y_1,\dots,y_{i}\}\\|C|=2}}\left(\hw(C)+d(C,y_i^{n-2})\right)$$
of the potential is  equal to
$\min_{x\in\{y_1,\dots,y_{i-1}\}}\left(\hw(y_i,x)+xy_i\right)$, since some
optimal $C$ contains $y_i$ due to the Lipschitz
property. Thus each term adds the
minimum-weight edge $[y_i,x]$ from point $y_i$ to the tree constructed so far.  To
show that $\Phi_{x_1,\ldots, x_k}(w)$ is minimized when $x_k=r$, it is equivalent
to show that $\hat \Phi_{y_1,\ldots,y_n}(w)$ is minimized when $y_n=r$, or
equivalently that $r$ is a leaf of some minimum spanning tree with weights
$\dw(x,y)+xy$.

\begin{lemma} \label{lemma:r-is-a-leaf}
	Suppose that $w\in \cW^{n-2}(r)$ and let $\dw(x,y)=w(M\setminus \{x,y\})$. Consider the set
	of spanning trees when the weights are $\dw(xy)+xy$. There exists a minimum
	spanning tree in which $r$ is a leaf.
\end{lemma}
\begin{proof}
	%	Note first that $\dw$ is a quasiconvex function: $\dw\in \cQ^2$.
	
	Take a minimum spanning tree $T$ with weights $\dw(xy)+xy$ in which $r$ has minimum
	degree. If $r$ is a leaf of this tree, there is nothing to show. Suppose otherwise
	that $r$ has degree at least two and we will reach a contradiction. Consider one of
	its neighbors, $x$. We have $\dw(r x)=\dw(y x)+ry$ for some $y\neq r,x$.
	\paragraph{Case 1:} The path in $T$ from $x$ to $y$ contains $r$.
	
	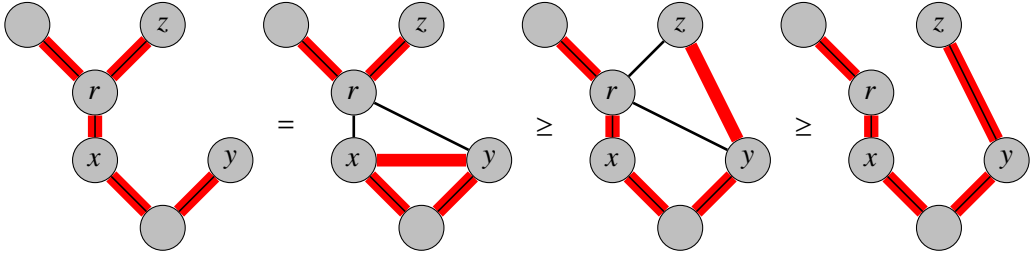
\begin{figure}
		\centering
		\begin{tikzpicture}[
		point/.style = {circle, draw, text centered, node distance=5cm, minimum size=17pt, inner sep=0pt, fill=black!25},scale=0.9]
		
		\node[point] at (3,0) (1) {};
		\node[point] at (2,1) (x) {$x$};
		\node[point] at (4,1) (y) {$y$};
		\node[point] at (2,2) (r) {$r$};
		\node[point] at (1,3) (2) {};
		\node[point] at (3,3) (z) {$z$};
		\draw[color=red,double=black,line width=2.4] (1) -- (x);
		\draw[color=red,double=black,line width=2.4] (1) -- (y);
		\draw[color=red,double=black,line width=2.4] (x) -- (r);
		\draw[color=red,double=black,line width=2.4] (r) -- (z);
		\draw[color=red,double=black,line width=2.4] (r) -- (2);
		\end{tikzpicture}
		\begin{tikzpicture}[
		point/.style = {circle, draw, text centered, node distance=5cm, minimum size=17pt, inner sep=0pt, fill=black!25},scale=0.9]
		
		\path (1,1.5) node {$=$};
		
		\node[point] at (3,0) (1) {};
		\node[point] at (2,1) (x) {$x$};
		\node[point] at (4,1) (y) {$y$};
		\node[point] at (2,2) (r) {$r$};
		\node[point] at (1,3) (2) {};
		\node[point] at (3,3) (z) {$z$};
		\draw[color=red,double=black,line width=2.4] (1) -- (x);
		\draw[color=red,double=black,line width=2.4] (1) -- (y);
		\draw[color=black,line width=1] (x) -- (r);
		\draw[color=black,line width=1] (y) -- (r);
		\draw[color=red,line width=4.8] (x) -- (y);
		\draw[color=red,double=black,line width=2.4] (r) -- (z);
		\draw[color=red,double=black,line width=2.4] (r) -- (2);
		\end{tikzpicture}
		\begin{tikzpicture}[
		point/.style = {circle, draw, text centered, node distance=5cm, minimum size=17pt, inner sep=0pt, fill=black!25},scale=0.9]
		
		\path (1,1.5) node {$\geq$};
		
		\node[point] at (3,0) (1) {};
		\node[point] at (2,1) (x) {$x$};
		\node[point] at (4,1) (y) {$y$};
		\node[point] at (2,2) (r) {$r$};
		\node[point] at (1,3) (2) {};
		\node[point] at (3,3) (z) {$z$};
		\draw[color=red,double=black,line width=2.4] (1) -- (x);
		\draw[color=red,double=black,line width=2.4] (1) -- (y);
		\draw[color=red,double=black,line width=2.4] (x) -- (r);
		\draw[color=black,line width=1] (r) -- (z);
		\draw[color=black,line width=1] (r) -- (y);
		\draw[color=red,line width=4.8] (y) -- (z);
		\draw[color=red,double=black,line width=2.4] (r) -- (2);
		\end{tikzpicture}
		\begin{tikzpicture}[
		point/.style = {circle, draw, text centered, node distance=5cm, minimum size=17pt, inner sep=0pt, fill=black!25},scale=0.9]
		
		\path (1,1.5) node {$\geq$};
		
		\node[point] at (3,0) (1) {};
		\node[point] at (2,1) (x) {$x$};
		\node[point] at (4,1) (y) {$y$};
		\node[point] at (2,2) (r) {$r$};
		\node[point] at (1,3) (2) {};
		\node[point] at (3,3) (z) {$z$};
		\draw[color=red,double=black,line width=2.4] (1) -- (x);
		\draw[color=red,double=black,line width=2.4] (1) -- (y);
		\draw[color=red,double=black,line width=2.4] (x) -- (r);
		\draw[color=red,double=black,line width=2.4] (y) -- (z);
		\draw[color=red,double=black,line width=2.4] (r) -- (2);
		\end{tikzpicture}
		\caption{Pictorial proof of Case 2a of Lemma~\ref{lemma:r-is-a-leaf}. Dark
			lines indicate distances and light (red) lines work function values.}
		\label{fig:evader}
	\end{figure}
	Then if we remove edge $[r,x]$ from tree $T$ and add edge $[x,y]$, we get a new
	spanning tree $T'$ with weight equal to the weight of $T$ minus $rx+ry-xy\geq
	0$. The new tree has weight at most equal to the weight of $T$ and the degree of
	$r$ is one less than its degree in $T$, a contradiction.
	
	\paragraph{Case 2:} The path in $T$ from $x$ to $y$ does not contain $r$.
	
	Since the degree of $r$ in $T$ is at least 2, let $z$ be another neighbor $r$.
	Then by the quasiconvexity property of $\dw$, one of the following two subcases must be true.
	
	\paragraph{Case 2a:} $\dw(r,z)+\dw(x,y)\geq \dw(r,x)+\dw(y,z)$.
	
	Using $\dw(r, x)=\dw(y, x)+ry$, we get $\dw(r,z)\ge \dw(y,z)+ry$. The tree
	$T'$ that we get if we replace edge $[r,z]$ with edge $[y,z]$ in $T$ has weight at most
	equal to the weight of $T$ minus $ry+rz-yz\geq 0$. As in the previous case, $T'$ is
	a minimum spanning tree in which the degree of $r$ is reduced by 1, a
	contradiction (see Figure~\ref{fig:evader}).

	\paragraph{Case 2b:} $\dw(r,z)+\dw(x,y)\geq \dw(r,y)+\dw(x,z)$.
	
	This is identical to the previous case when we interchange the role of $x$
	and $y$ (see Figure~\ref{fig:evader}).
	% In this case, we first replacing edge $[r,x]$ with edge $[r,y]$ in $T$, which yields a tree $\tilde T$ whose weight is equal to the weight of $T$ minus $\dw(xy)+rx-\dw(ry)\ge 0$. Since $T$ was a minimum spanning tree, this must hold with equality and $\tilde T$ is also a minimum spanning tree. Therefore, the arguments from Case 2a apply with $x$ and $y$ swapped.
\end{proof}

\subsection{Quasiconcavity and trees}
\label{sec:quasiconcavity-trees}

Let $T$ be a tree with set of leaves $M$ and with nonnegative weights on its
edges. We define its leaf-distance $v_T: M\times M\rightarrow \R$ such
that $v_T(x,y)$ is the weight of the unique path of $T$ between leaves $x$ and $y$.

\begin{theorem}
	A finite metric is quasiconcave if and only if it is a tree.
\end{theorem}
\begin{proof}
	We will show that a metric $v: M\times M\rightarrow \R$ is quasiconcave if and
	only if it is the leaf-distance of some tree $T$ with nonnegative weights.
	
	One direction is immediate: given a weighted tree $T$ with nonnegative weights,
	in every subtree with 4 leaves the maximum weighted matching between leaves is
	achieved in two out of the three possible matchings.
	
	For the opposite direction, given a quasiconcave metric $v$ we can create
	recursively a tree such that $v$ is equal to its leaf-distance function $v_T$, as
	follows: Create a tree $T'$ for the first $n-1$ elements of $M$ and extend it by
	adding a new leaf for the last element $z$ of $M$. To do this, we need to
	determine two things: first, a point $a$ of the tree where to attach the new edge
	$[a,z]$ and second the length of edge $[a,z]$. Since point $a$ splits some
	edge of $T'$ into two new edges and we add edge $[a,z]$, the resulting tree $T$
	has two more edges than $T$.
	
	To find point $a$, we first find $x$ and $y$ that minimize
	$v(x,z)+v(y,z)-v(x,y)$. We view the path from $x$ to $y$ in $T'$ as a real
	interval of length $v(x,y)$ and find a point $a$ on it with (nonnegative)
	distances $1/2(v(x,y)+v(x,z)-v(y,z))$, $1/2(v(x,y)+v(y,z)-v(x,z))$,
	$1/2(v(x,z)+v(y,z)-v(x,y))$, from $x$, $y$, and $z$ respectively. Note that this
	agrees with the distances between leaves $x$, $y$ and $z$.
	
	We need to verify that this construction works, i.e., that $v$ and $v_T$ are the
	same. The construction preserves the weight between leaves of $T'$, so we only
	need to verify that $v_T(z,u)=v(z,u)$, for every $u$. Note also that
	$v_T(x,z)=v(x,z)$ and $v_T(y,z)=v(y,z)$, by construction. Since $T$ is a tree,
	and $a$ is in the path between $x$ and $y$, we get
	$v_T(u,z)+v(x,y)=\max \{ v(x,z)+v(u,y), v(y,z)+v(u,x) \}$. On the other hand,
	since $x$ and $y$ were chosen to minimize $v(x,z)+v(y,z)-v(x,y)$, we must have
	that $v(u,z)+v(x,z)-v(u,x)\geq v(x,z)+v(y,z)-v(x,y)$ or equivalently,
	$v(u,z)+v(x,y)\geq v(y,z)+v(u,x)$. Interchanging the role of $x$ and $y$, we also
	have $v(u,z)+v(x,y)\geq v(x,z)+v(u,y)$. So, we must have that
	$v(u,z)+v(x,y)=\max \{ v(x,z)+v(u,y), v(y,z)+v(u,x) \}$ by quasiconcavity of $v$, which shows that
	$v_T(u,z)= v(u,z)$.
\end{proof}

\end{document}